\documentclass[journal]{IEEEtran}

\ifCLASSINFOpdf
   \usepackage[pdftex]{graphicx}
   \DeclareGraphicsExtensions{.pdf,.jpeg,.png}
\else
   \usepackage[dvips]{graphicx}
   \DeclareGraphicsExtensions{.eps}
\fi

\PassOptionsToPackage{bookmarks={false}}{hyperref}
\usepackage[draft]{hyperref}

\usepackage{amsthm}
\usepackage[cmex10]{amsmath}
\usepackage{graphicx}
\usepackage{caption}
\usepackage{subcaption}
\usepackage{color}
\usepackage{tikz}

\usepackage{setspace}
\usepackage{todonotes}
\usepackage{booktabs}

\usepackage{color}
\usepackage{amssymb}
\usepackage{amsthm}
\usepackage{pifont}

\usepackage{longtable}
\usepackage{supertabular}
\usepackage{blkarray}
\usepackage{setspace}
\usepackage{cite}
\usepackage{verbatim}

\usepackage{hyperref}
\usepackage{multirow}

\usepackage{algorithm}
\usepackage{algpseudocode}
\usepackage{booktabs}
\newtheorem{theorem}{Theorem}

\newtheorem{lemma}{Lemma}
\newtheorem{definition}{Definition}

\newcommand{\I}{\mathbb{I}}
\newcommand{\E}{\mathbb{E}}
\newcommand{\B}{\mathbb{B}}
\newcommand*\circled[1]{\tikz[baseline=(char.base)]{
            \node[shape=circle,draw,inner sep=2pt] (char) {#1};}}

\def\Verify{Verify}
\def\AddSamples{IncreaseSamples}
\def\IM{IM}

\def\CTVM{CTVM}

\def\BIM{BIM}

\def\BRIS{BSA}

\def\ISSA{\textsc{TipTop}}

\newcommand\R{\mathcal{R}}

\usepackage{siunitx}
\usepackage{array}
\newcolumntype{L}[1]{>{\raggedright\let\newline\\\arraybackslash\hspace{0pt}}m{#1}}
\newcolumntype{C}[1]{>{\centering\let\newline\\\arraybackslash\hspace{0pt}}m{#1}}
\newcolumntype{R}[1]{>{\raggedleft\let\newline\\\arraybackslash\hspace{0pt}}m{#1}}

\IEEEoverridecommandlockouts
\begin{document}
	\title{TipTop: (Almost) Exact Solutions for Influence Maximization in Billion-scale Networks}

    
\author{Xiang Li, J.~David~Smith, Thang N. Dinh, and~My~T.~Thai%
\thanks{Xiang Li is with the Department of Computer Engineering, Santa Clara University, Santa Clara, CA, 95053 USA. Email: xli8@scu.edu

J. David Smith and My T. Thai are with the Computer \& Information Science \& Engineering Department, University of Florida, Gainesville, FL, 32601 USA. E-mail: \{jdsmith, mythai\}@cise.ufl.edu

Thang N. Dinh is with the Computer Science Department, Virginia Commonwealth University, Richmond, VA 23284. Email: tndinh@vcu.edu

%My T. Thai is with Division of Algorithms and Technologies for Networks Analysis \& Faculty of Information Technology, Ton Duc Thang University, Ho Chi Minh City, Vietnam, and also with Department of Computer and Information Science and Engineering, University of Florida, Gainesville, FL, 32601 USA. Email: thaitramy@tdt.edu.vn; mythai@cise.ufl.edu

%My T. Thai is the corresponding author.

The preliminary version of this paper has appeared in \cite{li2017approximate}.
}
}
    \begin{comment}
    \author{
\IEEEauthorblockN{Xiang Li, J. David Smith}
\IEEEauthorblockA{CISE Department\\
University of Florida\\
Gainesville, FL, 32601\\
Email: \{xixiang, jdsmith\}@cise.ufl.edu}
\and
\IEEEauthorblockN{Thang N. Dinh}
\IEEEauthorblockA{CS\ Department\\
Virginia Commonwealth University\\
Richmond, VA 23284\\
Email: tndinh@vcu.edu}
\and
\IEEEauthorblockN{My T. Thai}
\IEEEauthorblockA{CISE\ Department\\
University of Florida\\
Gainesville, FL, 32601\\
Email: mythai@ufl.edu}
}
\end{comment}
\maketitle

\begin{abstract}
		%\boldmath
In this paper, we study the Cost-aware Target Viral Marketing (CTVM{}) problem, a generalization of Influence Maximization (\IM{}). \CTVM{} asks for the most cost-effective users to influence the most relevant users. In contrast to the vast literature, we attempt to offer exact solutions. As the problem is NP-hard, thus, exact solutions are intractable, we propose \ISSA{}, a $(1-\epsilon)$-optimal solution for arbitrary $\epsilon>0$ that scales to very large networks such as Twitter. At the heart of \ISSA{} lies an innovative technique that reduces the number of samples as much as possible. This allows us to exactly solve \CTVM{} on a much smaller space of generated samples using Integer Programming. Furthermore, \ISSA{} lends a tool for researchers to benchmark their solutions against the optimal one in large-scale networks, which is currently not available.

%We first highlight that sample average approximation (SAA) approaches with traditional two stage stochastic programming have excessive running time, even on networks of a few hundred of nodes. We then propose  \ISSA{}, a $(1 - \epsilon)$-optimal solution with arbitrary small $\epsilon$ that offer practically good solutions and scale well to networks. This result significantly improves the current best solutions to both \IM{} and \CTVM{}. At the heart of \ISSA{} lies an innovative technique that reduces the number of samples as much as possible. This allows us to exactly solve \CTVM{} on a much smaller space of generated samples using Integer Programming. While obtaining an almost exact solution, \ISSA{} is very scalable, running on billion-scale networks such as Twitter. Furthermore, \ISSA{} lends a tool for researchers to benchmark their solutions against the optimal one in large-scale networks, which is currently not available.
\end{abstract}
%With the developing of \ISSA{}, researchers can benchmark their solutions against the optimal one in large-scale networks, which is currently not available.


	\IEEEpeerreviewmaketitle

	\setlength{\intextsep}{0.5\baselineskip}

	\parskip 5pt

	\captionsetup{labelfont=bf}

	\vspace{-0.1in}
\section{Introduction}

With the recent development, Online Social Networks (OSNs) have become one of the most effective platforms for marketing and advertising. Through ``word-of-mouth'' exchanges, so-called viral marketing, the influence and product adoption can spread from few key users to billions of users in the network. To identify these key users, {Influence Maximization} (\IM{}) problem, which asks for a set of $k$ seed users that maximizes the expected number of influenced nodes, has been studied extensively \cite{Leskovec07, Chen10, Goyal11, Tang14, Tang15, Nguyen163} (and references therein). Taking into account both arbitrary cost for selecting a node and arbitrary benefit for influencing a node, a generalized problem, {Cost-aware Targeted Viral Marketing} (\CTVM{}), has been introduced recently \cite{Nguyen161}. Given a budget $\kappa$, \CTVM{} asks to find a seed set $S$ with the total cost at most $\kappa$ such as to maximize the expected total benefit over the influenced nodes.

%Since then several variants of \IM{} has been studied, such as \TVM{}, \BIM{} {\bf cite pls}.
%\CTVM{}  is more relevant in practice as it generalizes other viral marketing problems including  \TVM{}, \BIM{} and the fundamental \IM{}. However, the problem is much more challenging with heterogeneous costs and benefits. As we show in Section 3,  extending the state-of-the-art method for \IM{} in \cite{Tang14} may increase the running time by a factor $|V|$, making the method unbearable for large networks.
\begin{comment}
Solutions to both \IM{} and \CTVM{} either suffer from scalability issues or limited performance guarantees. Since \IM{} is NP-hard, Kempe et al. \cite{Kempe03} introduced the first $(1-1/e-\epsilon)-${\em approximation algorithm} for any $\epsilon >0$. However, it is not scalable. % as the Monte-Carlo sampling was used to assess the quality of solution at each iteration.
 This work has motivated a vast amount of work on \IM{} in the past decade \cite{Leskovec07, Chen10, Goyal11, Tang14, Tang15, Nguyen163} (and references therein), focusing on solving the scalability issue. The most recent noticeable results are IMM \cite{Tang15} and SSA \cite{Nguyen163} algorithms, both of which can run on large-scale networks with the same performance guarantee of  $(1-1/e-\epsilon)$. Meanwhile, BCT was introduced to solve \CTVM{} with an approximation ratio of $(1 - 1 /\sqrt{e} - \epsilon)$ \cite{Nguyen161}. 
%Goyal11_celf++, Cohen14, Ohsaka14

\end{comment}
Despite a great amount of works \cite{Kempe03, Leskovec07, Chen10, Goyal11, Tang14, Tang15, Nguyen163, Nguyen161, Borgs14,Nguyen2013}, none of these attempts to solve \IM{} or \CTVM{} exactly. The lack of such a solution makes it challenging to evaluate the performance of existing solutions, such as IMM \cite{Tang15} and SSA \cite{Nguyen163} algorithms, in real-world datasets, against optimal solutions. Despite the fact that these algorithms have a theoretical performance guarantee of $(1-1/e-\epsilon)$ in the worst case, one can always ask: how well do these algorithms actually perform on the billion-scale OSNs? This question has remained unanswered til now.

Obtaining exact solutions to \CTVM{} (and thus to \IM{}) indeed is very challenging. Due to the nature of the problem, stochastic programming is a viable approach for optimization under uncertainty when the probability distribution governs the data is given \cite{Shapiro14}. However traditional stochastic programming-based solutions to various problems in NP-hard class are only for small networks with a few hundreds nodes \cite{Shapiro14}. Thus directly applying existing techniques to \IM{} and \CTVM{} is not suitable because OSNs consist of millions of users and billions of edges. Furthermore, the theory developed to assess the solution quality such as those in \cite{Bayraksan06, Shapiro14}, and the references therein, only provide approximate confidence interval. Therefore, sufficiently large samples are needed to justify the quality assessment. This requires us to develop novel stochastic programming techniques to optimally solve \CTVM{}.

In this paper, we provide the first (almost) exact solutions for \CTVM{} with an approximation ratio of $(1 - \epsilon)$. To tackle the above challenges, we develop two innovative techniques: 1) Reduce the number of samples as much as possible so that the stochastic programming can be solved in a short time. This requires us to tightly bound the number of samples needed to generate a candidate solution. 2) Develop novel computational method to assess the solution quality with just enough samples, where the quality requirement is given a priori. These results cross the barriers in stochastic programming theory where solving stochastic programming on large-scale networks had been thought impractical. Our contributions are summarized as follows:
%	In this paper, we propose the computation of \emph{influence spectrum} (\IS), the maximum influence (and the corresponding seed sets) for all possible seed sizes from $k=1$ up to $n$. The influence spectrum gives better insights for decision making and resource planning in viral marketing campaigns. Given the influence spectrum, we can find the solutions for not only \IM{} but also cost-effective seed set  \cite{Leskovec07} and  min-seed set selection \cite{Long11} problems (with the best approximation guarantees). As useful as it is, no one has ever considered computing \IS{} due to the perception that it seems extremely computational expensive. The fact is computing \IS{} implies solving of $n$ \IM{} instances with seed size as large as $n$. Unfortunately,  \IM{} methods either do not scale well with large seed sets \cite{Kempe03, Leskovec07, Goyal11, Goyal11_celf++} or resort to heuristics \cite{Chen10}, i.e., obtained results coulld be arbitrarily worse than the optimal ones. One might look into adapting some greedy methods for \IM{} for the task, e.g., solving \IM{} with $k=n$ and output the solutions for all intermediate values of $k=1..(n-1)$.  When attempting to adapt the state-of-the-art method for \IM{}  in \cite{Tang14} for \IS{}, we obtain a prohibitive high time complexity $O(n(m+n) \log n \epsilon^{-2})$, rendering the algorithm unsuitable for the task.
	\begin{itemize}
		\item Design two exact solutions, namely T-EXACT and E-EXACT,  to \CTVM{} using two-stage stochastic programming which utilizes the sample average approximation method to reduce the number of realizations. Using T-EXACT and E-EXACT, we illustrate that traditional stochastic programming techniques badly suffer the scalability issue. 
       % \vspace{-5pt}
        \item Develop an (almost) optimal algorithm to \CTVM{} with a performance ratio of $(1 - \epsilon)$: The Tiny Integer Program with Theoretically OPtimal results (\ISSA{}). Being able to obtain the optimal solution, \ISSA{} is used as a benchmark to evaluate the absolute performance of existing solutions in billion-scale OSNs.
     %   \vspace{-5pt}
        \item Conduct extensive experiments confirming that the theoretical performance of \ISSA{} is attained in practice. Our experiments show that it is feasible to compute 98\% optimal solutions to \CTVM{} on billion-scale OSNs. These experiments confirm that our sampling reductions are significant in practice, by a factor of $10^3$ on average.
	\end{itemize}

{\bf Organization}. Section \ref{sec:rel} briefly discusses the related work and the Reverse Influence Sampling (RIS). In Section \ref{sec:model}, we present the network model, propagation models, and the problem definition. Section \ref{sec:stoc} presents our EXACT algorithms for \CTVM. Our main contribution, \ISSA{} is introduced in Section \ref{sec:tiptop}. We analyze \ISSA{} approximation factor in Section \ref{sec:complexity}. Experimental results on real social networks are shown in Section \ref{sec:exp}. And finally Section VIII concludes the paper.
%	\setlength{\tabcolsep}{4pt}
%	\renewcommand{\arraystretch}{1.2}
	
%	\begin{table}[htb]\small
%		%\vspace{-0.1in}
%		\centering
%		\caption{Summary of the problems and methods in the literature of Influence Maximization and Viral Marketing}
%		\begin{tabular}{l|C{0.8cm}|C{0.8cm}|C{0.8cm}|C{1cm}}
%			\addlinespace
%			\bf Name  &  \bf \IM{} & \bf \BIM{} & \bf \TVM & \bf \CTVM \\
%			\midrule
%			Greedy (\cite{Kempe03}) & \checkmark & & & \\
%			RIS (\cite{Borgs14}) & \checkmark & & & \\
%			TIM/TIM+ (\cite{Tang14}) & \checkmark & & & \\
%			BIM (\cite{Nguyen2013}) & & \checkmark & & \\
%			INFLEX (\cite{Aslay2014}) & & & \checkmark & \\    
%			\CTIM{} (this paper) & \cmark & \cmark & \cmark & \cmark \\
%			\bottomrule
%		\end{tabular}%
%		\label{tab:time_complexity}%
%		%  \vspace{-0.28in}
%	\end{table}% 

\section{Related work and Reverse Influence Sampling}\label{sec:rel}

%Since then several variants of \IM{} has been studied, such as \TVM{}, \BIM{} {\bf cite pls}.
%\CTVM{}  is more relevant in practice as it generalizes other viral marketing problems including  \TVM{}, \BIM{} and the fundamental \IM{}. However, the problem is much more challenging with heterogeneous costs and benefits. As we show in Section 3,  extending the state-of-the-art method for \IM{} in \cite{Tang14} may increase the running time by a factor $|V|$, making the method unbearable for large networks.
\begin{comment}
Solutions to both \IM{} and \CTVM{} either suffer from scalability issues or limited performance guarantees. Since \IM{} is NP-hard, Kempe et al. \cite{Kempe03} introduced the first $(1-1/e-\epsilon)-${\em approximation algorithm} for any $\epsilon >0$. However, it is not scalable. % as the Monte-Carlo sampling was used to assess the quality of solution at each iteration.
 This work has motivated a vast amount of work on \IM{} in the past decade \cite{Leskovec07, Chen10, Goyal11, Tang14, Tang15, Nguyen163} (and references therein), focusing on solving the scalability issue. The most recent noticeable results are IMM \cite{Tang15} and SSA \cite{Nguyen163} algorithms, both of which can run on large-scale networks with the same performance guarantee of  $(1-1/e-\epsilon)$. Meanwhile, BCT was introduced to solve \CTVM{} with an approximation ratio of $(1 - 1 /\sqrt{e} - \epsilon)$ \cite{Nguyen161}. 
%Goyal11_celf++, Cohen14, Ohsaka14
\end{comment}
{\bf Influence Maximization.}	Kempe {\em et al.} \cite{Kempe03} formulated viral marketing as the \IM{} optimization problem, focused on two fundamental cascade models, Linear Threshold (LT) and Independent Cascade (IC) models. They showed the problem to be NP-complete and devised an $(1-1/e - \epsilon)$\ approximation algorithm.  In addition, \IM{} cannot be approximated within a factor $(1-\frac{1}{e}+\epsilon)$ \cite{Feige98} under a typical complexity assumption. Computing the exact influence is shown to be \#P-hard \cite{Chen10}.
    
    Following \cite{Kempe03}, a series of work have been proposed, focused on improving the time complexity \cite{Leskovec07, Chen10, Goyal11, Goyal11_celf++,Cohen14, dinh2014approximability, Ohsaka14, Tang14, HeJing}. All of these work retain the ratio of $(1-1/e - \epsilon)$.
    %A notable one is Leskovec et al \cite{Leskovec07}, in which the lazy-forward heuristic (CELF) was used to obtain about 700-fold speed up.
    Their major bottle-neck is the inefficiency in estimating the influence spread, thus restraining them from being able to run on large networks.% All of these work retain the ratio of $(1-1/e - \epsilon)$, and are not able to run on large-scale networks, mainly due to an inefficient sampling technique of Monte Carlo. 

	%Several heuristics are developed to derive solutions in large networks. While those heuristics are often faster in practice, they fail to retain  the $(1-1/e -\epsilon)$-approximation guarantee and produce lower quality seed sets. Chen et al. \cite{Chen09} obtain a speed up by using an influence estimation for the IC model. For the LT model, Chen et al. \cite{Chen10} propose to use local directed acyclic graphs (LDAG) to approximate the influence regions of nodes. In a complement direction, there are works on learning the parameters of influence propagation models \cite{Goyal10learning, Kutzkov13}. 
    
{\bf Reverse Influence Sampling (RIS).} Borgs {\em et al.} \cite{Borgs14} have introduced a novel sampling approach, RIS, which is a foundation for the later works. Briefly, RIS captures the influence landscape of $G = (V, E, p)$ through generating a hypergraph $\mathcal H=(V, \{ \mathcal E_1, \mathcal E_2,\ldots\})$. Each hyperedge $\mathcal E_j \in \mathcal H$ is a subset of nodes in $V$ and constructed as follows: 1) selecting a random node $v \in V$ 2) generating a sample graph $g \sqsubseteq G$ and 3) returning $\mathcal E_j$ as the set of nodes that can reach $v$ in $g$. Observe that  $\mathcal E_j$ contains the nodes that can influence its source $v$. If we generate multiple random hyperedges, influential nodes will likely appear more often in the hyperedges. Thus a seed set $S$ that \emph{covers} most of the hyperedges will likely maximize the influence spread. Here $S$ covers a hyperedge $\mathcal E_j$, if $S \cap \mathcal E_j \neq \emptyset$. Therefore, \IM{} can be solved using the following framework. 1) Generate multiple random hyperedges from $G$. 2) Use the greedy algorithm for the Max-coverage problem \cite{Khuller1999} to find $S$ that covers the maximum number of hyperedges and return $S$ as the solution. The core issue in applying the above framework is that: How many hyperedges are sufficient to provide a good approximation solution?

%   Unfortunately, computing $OPT_k^{IM}$ is intractable, thus, TIM/TIM+ in \cite{Tang14}  have to generate $\theta \frac{OPT_k^{IM}}{KPT^+}$ hyperedges, where the ratio $\frac{OPT_k^{IM}}{KPT^+} \geq 1$ is not upper-bounded. That is TIM/TIM+ may  generate many times more  hyperedges than needed. In contrast, our BCT algorithm in Section IV guarantees that the number of hyperedges is at most a constant time of the theoretical threshold (with high probability). Thus, its running time is smaller and more predictable.
    Based on RIS, Borgs {\em et al.} \cite{Borgs14} presented an $O(kl^2(m+n)\log^2 n/\epsilon^3)$ time algorithm for \IM{} under IC model. It returns a $(1-1/e - \epsilon)$-approximate ratio with probability at least $1 - n^{-l}$. In practice, the proposed algorithm is, however, less than satisfactory due to the rather large hidden constants. In a sequential work, Tang et al. \cite{Tang15} reduced the running time to $O( (k+l) (m+n) \log n/\epsilon^2)$ and showed that their algorithm is efficient in billion-scale networks. Nguyen et al. \cite{Nguyen163} proposed SSA/DSSA algorithms to further reduce the running time up to orders of magnitudes. The algorithm keeps generating samples and stops at exponential check points to verify (stare) if there is adequate statistical evidence on the solution quality for termination. Huang et al. showed gaps  in SSA/D-SSA \cite{huang2017revisiting} and 
    propose the fixes for SSA. Independently, the authors of SSA/D-SSA provided the fixes for both SSA and D-SSA in \cite{sigmodarxiv16}  with the summary of changes in \cite{nguyen2018revisiting}. However, SSA does not put an effort in minimizing the number of samples at the stopping point to verify the candidate solution, which is needed to solve the Integer Programming (IP). All of these works have a $(1-1/e - \epsilon)$-approximation ratio.

 %   However, IMM has two weaknesses: 1) intractable estimation of maximum influence and 2) taking union bounds over all possible seed sets in order to guarantee a single returned set. 
     
 {\bf Generalization.}  In generalizing \IM{}, Nguyen and Zheng \cite{Nguyen2013} investigated the \BIM{} problem in which each node can have an arbitrary selecting cost. They proposed a $(1-1/\sqrt{e} -\epsilon)$ approximation algorithm (called BIM) based on a greedy algorithm for Budgeted Max-Coverage in \cite{Khuller1999} and two heuristics. However, none of the proposed algorithms can handle billion-scale networks. Recently, Nguyen {\it et. al} introduced \CTVM{} and presented a scalable $(1-1/\sqrt{e} -\epsilon)$ algorithm \cite{Nguyen161}. They also showed that straightforward adaption of the methods in \cite{Borgs14, Tang14, Tang15} for \CTVM{} can incur an excessive number of samples, thus, are not efficient enough for large networks.
	\section{Models and Problem Definitions}\label{sec:model}

%	In this section, we present the \CTVM{} problem and  an overview of the Reverse Influence Sampling approaches in Borgs et al. \cite{Borgs14} and Tang et al. \cite{Tang14}. We focus on  the \emph{Independent Cascade} (IC) model and  the \emph{Linear Threshold} (LT) in \cite{Kempe03}. For simplicity, we focus on the IC model and discuss the extension to LT.
%\vspace{-0.1in}
%\subsection{Model and Problem Definition}	
Let $\mathcal G=(V, E, c, b, \mathcal D)$ be a network with a node set $V$ and a directed edge set $E$, with $|V| =n$  and $|E|=m$. Each node $u \in V$ has a selecting cost $c_u\geq 0$, also written $c(u)$, and a benefit $b(u)$ if $u$ is influenced\footnote{The cost of node $u$, $c_u$, can be estimated proportionally to the centrality of $u$ (how important the respective person is), e.g., out-degree of $u$ \cite{Nguyen2013}. Additionally, the node benefit $b(u)$ refers to the gain of influencing node $u$, e.g., 1 for each node in our targeted group and 0 outside \cite{Chen2015}.}. Each directed edge $(u, v) \in E$ is associated with an influence probability $p_{uv} \in [0, 1]$.

The diffusion of information in $\mathcal G$ is captured through a probability space $\mathcal D$ that associates each possible cascade graph $g=(V, E_s), E_s \subseteq E$ with a probability. Each cascade graph $g$ is also called  a \emph{realization} or a \emph{sample graph} of $\mathcal G$. An edge $(u, v)\in E_s$ in $g$ implies that $u$ can activate $v$ in that graph. 

Given a seed set $S \subset V$, the number of nodes get influenced by $S$ within a sample graph $g$ is defined as the number of nodes reachable from $S$ in $g$, and denoted by $R(g, S)$. The \emph{influence spread} of $S$ in $\mathcal{G}$ is, similarly, defined as the expected influence of $S$ over all possible realization graphs in $\mathcal D$. Mathematically, define $\mathcal D = (\Omega, \mathcal F, P)$ where $\Omega = \{ g=(V, E_s) | E_s \subseteq E \}$ is the set of all possible graph samples of $\mathcal G$, $\mathcal F = 2^\Omega$, and $P: \mathcal F \rightarrow [0, 1]$, a probability measure. 

Let $G$ denote a \emph{random graph} defined over $\mathcal D$. The influence spread of the seed set $S$ is
\begin{align}
    	\I(S) = \E[ R(G, S)] =  \sum_{g  \in \Omega } P(g) |R(g, S)|,
\end{align}	
where $R(g, S)$ denotes the set of nodes reachable from $S$ within $g$.

	 For example, we consider  the popular Independent Cascade (IC) model \cite{Kempe03}. In IC, the influence propagation happens in round $t=1, 2, 3, \ldots$. At round 1, nodes in $S$ are \emph{activated} and the other nodes are \emph{inactive}. The cost of activating $S$ is given $c(S) = \sum_{u \in S} c_u$. At round $t>1$, each newly activated node $u$ will independently activate its neighbor $v$ with a probability $p_{uv}$. Once a node becomes activated, it remains activated in all subsequent rounds. The influence propagation stops when no more nodes are activated.
    
    For IC, the probability mass function for each sample graph  $g=(V, E_s)$ is
\begin{align*}
 P(g) = \textsf{Pr}[G= g] = \displaystyle\prod_{e \in E_s}p_e \displaystyle\prod_{e \in E\setminus E_s}(1-p_e).
  \vspace{-7pt}
\end{align*}

	Similarly, the \emph{benefit} of $S$ is defined as the expected total benefit over all influenced nodes, i.e., 
	\begin{align}
		\B(S) = \E[B(G, S)] = \sum_{g  \in  \Omega} P(g) B(g, S), 
	\end{align}
where $B(g, S) = \sum_{u \in R(g, S)} b(u)$ is the benefit of selecting $S$ with respect to (w.r.t.) graph sample $g$. 

Without loss of generality (w.l.o.g.), we assume that the benefit of the nodes are normalized so that: %$\sigma_B = \sum_{u \in V} b(u) = n$. This is compatible with the uniform benefit case in the IM problem in which each node has a same benefit one and $B(g, S) = R(g, S)$.
\begin{align}
\label{eq:sumb}
\sigma_B = \sum_{u \in V} b(u) = n.
\end{align}
This is compatible with the uniform benefit case in the IM problem in which each node has a same benefit one and $B(g, S) = R(g, S)$.

	\noindent We are now ready to define the \CTVM{} problem as follows.
	
	\begin{definition}[Cost-aware Targeted Viral Marketing - \CTVM{}]
 		Given  a graph and its diffusion model $\mathcal G=(V,E, c, b, \mathcal D)$  and a budget $\kappa>0$, find a seed set $S \subset V$ with the total cost $c(S) \leq \kappa$ to maximize the  benefit $\B(S)$.
	\end{definition}

\vspace{-0.1in}
\section{Multi-stage Programming Formulations}\label{sec:stoc}
In this section, we present two variants of EXACT solutions, T-EXACT and E-EXACT and their corresponding sample average approximation T-SAA and E-SAA, respectively. T-EXACT is based on two-stage stochastic programming while E-EXACT is the edge-based formulation with cycle-elimination. We then discuss the scalability issue of traditional stochastic programming which is later verified in our experiment, shown in Section \ref{sec:exp}. 

\vspace{-0.1in}
\subsection{Two-stage Stochastic Linear Program (T-EXACT)}
\label{subsec:formula} 

Given an instance  $\mathcal G=(V,E, c, b, p)$ of \CTVM{}, we first use integer variables $s_v, v \in V$ 
 to represent whether or not node $v$ is selected as a seed node. That is, $s_v = 1$ if $v$ is selected; otherwise $s_v = 0$.
% \begin{align}
 %	s_v = \left\{\begin{array}{rl} 1 & \text{if $v$ is selected, }\\ 0 & %\text{otherwise.} \end{array}\right. 
 % \end{align} 
  
  %We assume nodes are numbered from $1$ to $n = |V|$.
Variables $s$ are known as first stage variables. The values of $s$ are to be decided before the actual realization of the uncertain parameters in  $G$. Further, the set of selected nodes need to satisfy the budget constraint 
$
\sum_{v \in V} s_v c_v  \leq \kappa.
$ 

%We associate with each node pair $(u, v)$ a random Bernoulli variable $\xi_{uv}$ satisfying $\textsf{Pr}[\xi_{uv}=1] = p_{uv}$ and $\textsf{Pr}[\xi_{uv}=0] = 1 -p_{uv}$. For each realization of $\mathcal G$, the values of $\xi_{uv}$ are revealed to be either $0$ or $1$ and we can compute the benefit of the seed set $S$ implied by $s$. 

Given a random graph $G$, we define variable $x_{v}, v \in V,$ to be the activation state of node $v$ when the propagation stops. That is, $x_v = 1$ if $v$ is eventually activated; otherwise $x_v = 0$.
%$
%x_{v} =
% \left\{\begin{array}{rl} 1 & \text{if $v$ is eventually activated, }\\ 0 & %\text{otherwise.}
%\end{array}\right.
%$

The benefit obtained by seed set can be computed using a second stage mixed integer programming, denoted by $B(s, x, G)$ as follows. 
\begin{align}      
  \label{st:obj} B(s, x, G)  =\max& \displaystyle\sum_{v\in V}  b(v) x_v \\       
\mbox{s. t.\ }   &\label{st:seed} \sum_{u \in IR(G, v)} s_{u} \geq x_v,\   v \in V,\\
&\label{st:range}s_u \in \{0,1\}, x_{v} \in [0, 1]
\end{align}
where $IR(G, v)$ denotes the set of nodes $u$ so that there exists a path from $u$ to $v$  in $G$.

%Note that we do not require variables $x_v$ to be integral. Assume that we obtain an optimal solution with a fractional variable $x_v \in (0, 1)$ for some $v$, we can simply round up $x_v$ to $1$, thus, increase the objective. None of the constraints will be violated since $x_v$ only appears in the constraints (\ref{st:seed}). For those constraints, the fact that $x_v >0$ implies that some of the $s_u$ on the left-hand side must be one. Thus, the constraint still holds after rounding up $x_v$ to one.

The two-stage stochastic linear formulation for the CTVM problem is as follows.
 \begin{align}   
    \max_{s \in \{0, 1\}^n }\ \ \,  &    \mathbb{E}\left[ B(s, x, G) \right] \\    
     \mbox{s. t.\,} & \sum_{v\in V} s_v c_v  \leq \kappa  \\
     &\mbox{where }B(s, x, G) \text{ is given in (\ref{st:obj})-(\ref{st:range}) }
\end{align}

The objective is to maximize the expected benefit of the activated nodes $\mathbb{E}\left[ B(s, x, G) \right]$, where $B(s, x, G)$ is the optimal value of the second-stage problem. This stochastic programming problem is, however, not yet ready to be solved with a linear algebra solver.

\subsubsection{Discretization} To solve a two-stage stochastic problem, one often needs to discretize the problem into a single (very large) linear programming problem.  That is we need to consider all possible realizations $g \in \Omega$ and their probability masses $\textsf{Pr}[G=g]$. The  two-stage stochastic program can be discretized into a mixed integer programming, denoted by MIP$_F$ as follows. 

\begin{small}
\begin{align}   
    \label{st2:obj}\max\,  &  \sum_{g \in \Omega}  \textsf{Pr}[G=g]\sum_{v} b(v)x_v^l \\    
\mbox{s. t.\ }  &\label{st2:seed}  \sum_{v\in V} s_v c_v  \leq \kappa  \\
&\label{st2:activation}\sum_{u \in IR(g, v)} s_{u} \geq x_v^g,\   v \in V, g \in  \Omega \\
&\label{st2:range}s_u \in \{0,1\}, x_{v}^g \in [0, 1]
\end{align}
\end{small}
\vspace{-0.15in}
\subsection{Sample Average Approximation T-SAA}
An approach to reduce the number of realizations in T-EXACT is to apply the Sample Average Approximation (SAA) method. In that method, we generate independently $T$ graph samples $G^1,G^2,\cdots,G^T$ from $\mathcal D$. 

The expectation objective $q(s)=\mathbb E[B(s, x, G)]$ is then approximated by the sample average $\hat{q}_T(x) = \frac{1}{T} \sum_{l=1}^T \sum_{v}(b(v)x_{v}^l)$, and the new formulation is then
\begin{align}
\label{eq:texactob}\max\quad\quad\quad & \frac{1}{T} \sum_{l=1}^T  \sum_{v}b(v)x_{v}^l & \\ 
\nonumber \text{s. t.} \quad\quad\quad& \mbox{Constraints } (\ref{st2:seed})-(\ref{st2:range}),
\end{align}
where $x^l_v \in [0, 1]$ is an abbreviation for $x^{G^l}_v$.

We shall refer to the above mixed integer linear programming as T-SAA.

% Under some regularity conditions $\frac{1}{T} \sum_{j=1}^T \sum_{i<j}(1-x_{ij}^l)$ converges pointwise with probability $1$ to $\mathbb E[B(s, x, \xi)]$ as $T \rightarrow \infty$. Moreover, an optimal solution of the sample average approximation provides an optimal solution of the stochastic programming with probability approaching one exponentially fast w.r.t. $T$. Formally, denote by $s^*$ and $\hat s$ the optimal solution of the stochastic programming and the sample average approximation, respectively. For any $\epsilon >0$, it can be derived from Proposition 2.2 in \cite{Kleywegt02}  that
% \begin{align}
% \nonumber\textsf{Pr}\left[\mathbb{E}\left[ B(s, \hat x, \xi) \right]  - \mathbb{E}\left[ B(s, x^*, \xi) \right] > \epsilon \right]\\
%  \leq \exp\left( -T \frac{\epsilon^2}{n^4} + k\log n \right)
% \end{align}
%  Equivalently, if $T \geq \frac{n^4}{\epsilon^2} (k\log n - \log \alpha)$,  
%  then $\textsf{Pr}\left[\mathbb{E}\left[ B(s, \hat x, \xi)\right] - 
%  \mathbb{E}\left[ B(s, x^*, \xi) \right] < \epsilon\right] > 1 - 
%  \alpha$ for any $\alpha \in (0, 1)$. 

We start with identifying the number of samples $T$ needed to guarantee an $\epsilon$ error, followed by the expected (exponential) time complexity to solve the above mixed integer linear programming with $T$ samples.

\subsubsection{Sample Complexity} We bound the concentration with Hoeffding's inequality

\begin{lemma}[Hoeffding's inequality]
Let $X_1,\ldots,X_n$ be independent random variables in $[0, 1]$.  Let $\bar{X} = \frac{1}{n}\sum_{i=1}^n X_i$. Then, we have $\Pr[|\bar{X} - X| > t] \leq 2 \exp{\left(-2nt^2\right)}$.
\end{lemma}

%Let \begin{align}
%\mathcal C = \{ S \subseteq V : \sum_{v\in S} {s_v c_v} \leq \kappa \}
%\end{align}
Let $\mathcal C = \{ S \subseteq V : \sum_{v\in S} {s_v c_v} \leq \kappa \}$ be the set of all candidate seed sets.  In the worst-case, $\mathcal C$ can contains exponentially many candidates. 

The following lemma gives a bound on the number of necessary samples to guarantee an $\epsilon \sigma_B$ \emph{additive error}.% where $\sigma_B$ is the total benefit defined in Eq.~\ref{eq:sumb}.
\begin{lemma}
\label{lem:texactsamp}
	For fixed $T= \Omega(\frac{1}{\epsilon^2} \log \frac{2}{\delta} \log |\mathcal C|)$,  we have     
    \[
    \Pr[|\bar B(S) - \B(S)| > \epsilon \sigma_B] \leq \delta,
    \]
%  where $\sigma_B$ is the total benefit of all nodes, defined in Eq.~\ref{eq:sumb}.
\end{lemma}
\begin{proof}
For a fixed candidate solution $S \in \mathcal C$, apply the Hoeffding's inequality on $\frac{1}{\sigma_B} B(G^1, S), \frac{1}{\sigma_B} B(G^2, S),\ldots, \frac{1}{\sigma_B} B(G^T, S)$, we have
\[
	\Pr[|\bar B(S) - \B(S)| > \epsilon \sigma_B] \leq 2 \exp\left(-2\epsilon^2  T\right).
\]
By choosing $T=\Omega\left(\frac{1}{\epsilon^2} \log \frac{2}{\delta} \log|\mathcal C|\right)$ and taking the union bound over all candidate solutions in $\mathcal C$, we obtain the desired error bound.
\end{proof}

In the worst-case, $\mathcal C = 2^V$, the powerset of $V$, thus, $|\mathcal C| = 2^n$ and $T=O(n 1/\epsilon^2)$.

The estimation on $T$ maybe too conservative for practical estimates. While it provides some evidence on the convergence of the solution, it is excessively large for practical purposes.
\subsubsection{Time complexity} We analyze the time complexity of T-SAA in Eq. \ref{eq:texactob}, assuming an exhaustive search on $0-1$ integer variables $s_v, v\in V$. While the branch-and-cut (and other mixed integer linear programming methods) performs much better in practice, it has the same worst-case time complexity. 

For each of the $2^n$ possible assignments of $s_v$, we need to solve a remaining linear programming of $n T$ random variables $x^l_v$ of size, measured by the number of non-zeros, $M$. The expectation of $M$ depends on $T$ and the expected influence of nodes in $G$ as characterized in the following lemma.

\begin{lemma}
\label{lem:tsize}
The expected size, measured by the number of non-zeros, of T-SAA is $\E[M] = T \times \sum_{u \in V} I(u)$, where $I(u)$ denotes the expected influence of node $u\in V$. 
\end{lemma}
\begin{proof}
For each $u \in V$, the number of constraints (\ref{st2:activation}) that $u$ was on the left hand sides equal the number of nodes that $u$ can reach to. Thus, the expected number of constraints (\ref{st2:activation}) that $u$ participates into is, hence, $I(u)$. Taking the sum over all possible $u \in V$, we have the expected size of the T-EXACT with $T$ graph realizations is $T \times \sum_{u \in V} I(u)$.
\end{proof}

In other words, for each $G^l$, the size of T-SAA increases by the size of transitive closure in $G^l$, i.e., the number of pairs $(u, v)$ that $u$ can reach to $v$. For many subgraphs, this increase is of $O(n^2)$, making T-SAA bloats rapidly as $T$ increases.

To bound the time complexity of solving the LP, we use the following time bound on Karmarkar's algorithm 
\begin{theorem}{\cite{karmarkar1984new}}
\label{theo:karmarkar}
Denoting by N the number of variables and L the number of bits of input in a linear programming, then the runtime of Karmarkar's algorithm is 
\[
O(N^{{3.5}}L^{2}\cdot \log L\cdot \log \log L).
\]
\end{theorem}

Substitute $N = nT=O(\frac{1}{\epsilon^2 }n^2)$ and $L=O(M) = O(T \times \sum_{u \in V} I(u))=O(\frac{1}{\epsilon^2} n^3)$ from Lemma~\ref{lem:tsize}, we obtain the approximate time to solve T-SAA as
\[
O\left( 2^n  n^7 1/{\epsilon^7} 1/{\epsilon^4} n^6 polylog(n)   \right)
\]
Simplify and we get the following result
\begin{lemma}
The worst-case time complexity of T-SAA is
$O\left(2^n n^{13} \frac{1}{\epsilon^{11}} polylog(n) \right)$.
\end{lemma}

\vspace{-0.1in}
\subsection{Edge-based formulation with cycle-elimination (E-EXACT)}\label{sec:edge-ilp}
We provide an alternative formula in which for each random graph $G$, the size of the mixed integer linear program increases by $O(|E|)$ rather than $O(n^2)$ as in T-EXACT.

It is important that we formulate the second stage as a maximum problem, so that we will obtain  a bi-level $\mathbf{Max-max}$ optimization. The alternative $Max-min$ formulation is not only more sophisticated but also constrained to small size instances in practice. 

The main idea is \emph{to build a cascade tree from the seed nodes and count the number of activated nodes} instead of counting the number of activated nodes like in T-EXACT.
Given a random graph $G = (V, E)$, for each $(u, v) \in E$ we define a variable 
$
y_{uv}=
\left\{
\begin{array}{rl}
1& \text{if the edge $(u, v)$ is ``active'',}\\
0& \text{otherwise.}
\end{array}
\right.
$

 To ``build'' a cascade tree, we constraint that each node $v$ in $V$ has at most one active edge $(u, v)$ going to $v$, as shown in Eq.~(\ref{ipcycle:tree}).
 In addition, $(u, v)$ is active if and only if a) $(u, v)$ is an edge on the graph realization and b) either $u$ is selected, i.e., $s_u = 1$, or there exists active edge $(w, u)$ going to $u$, see Eq.~(\ref{ipcycle:seed}). Moreover, we forbid cycles composed of all active edges. 
This is similar to the sub-tour elimination for the TSP problem \cite{Dantzig54solutionof}. There might be an exponential number of cycles, however, the cycle can be added gradually. In each step, we identify a cycle of which constraint is violated and add the constraint to the programming formulation. Given a fractional solution $(s; y)$, an exact separation algorithm for some class of inequalities either finds a member of the class violated by $(s; y)$ or proves that no such member
exists. There is an exact algorithm for the separation procedure based on finding the shortest path as follow.

Let $z_{uv} = 1 - y_{uv}$, the constraint (\ref{ipcycle:cycle}) can be rewritten as $\sum_{(u, v) \in C} z_{uv} \geq 1$
%\[
%\sum_{(u, v) \in C} z_{uv} \geq 1.
%\] 
Thus we can find violated constraint by looking for the smallest length cycles in the graph with edges' lengths $z_{uv}$. In that graph, each edge $(u, v)$ with $z_{u,v} + d(u, v) < 1$, where $d(u,v)$ denotes the shortest distance between $u$ and $v$, will correspond to an violated constraint. The major time complexity in finding the violated cycles is on finding all-pair-shortest paths which can be solved in $O(n^2 \log n + n m)$ using the Johnson's algorithm.

Since each activated node $u$ is either already in the seed set or activated by exactly one neighbor in the built cascade tree, the objective and the complete formulation is then as follows.
\begin{align}      
 \nonumber B_E(s, x, G)  &=\\
  \label{ipcycle:obj} \max& \displaystyle\sum_{ (u, v)\in E} b(v) y_{uv} +\sum_{u\in V} b(u)s_u\\       
\mbox{s. t.\ }   &\label{ipcycle:seed} \sum_{w \in N^-(u)} y_{wu}+ s_{u} \geq y_{uv},\   (u, v) \in E\\
&\label{ipcycle:tree} \sum_{u \in N^-(v)} y_{uv}+s_v \leq 1 ,\  v\in V \\
&\label{ipcycle:cycle} \sum_{(u,v) \in C} y_{uv} \leq |C|-1 ,\ \text{ any cycle } C\\
&s_u \in \{0,1\},\  u \in V,\\
&\label{ipcycle:range} y_{uv} \in [0, 1],\  (u, v) \in E
\end{align}

The two-stage stochastic linear formulation for the CTVM problem is as follows.
 \begin{align}   
    \max_{s \in \{0, 1\}^n }\ \ \,  &    \mathbb{E}\left[ B_E(s, x, G) \right] \\    
     \mbox{s. t.\,} & \sum_{v\in V} s_v c_v  \leq \kappa  \\
     &\mbox{where }B_E(s, x, G) \text{ is given in (\ref{ipcycle:obj})-(\ref{ipcycle:range}) }
  \end{align}

 The following lemma proves the one-to-one mapping between activated nodes (not in the seed) and the active edges in the cascade tree.

\begin{lemma}
The number of activated nodes will be the sum of the number of active edges plus the number of seed nodes.
\end{lemma}
\begin{proof}
Define $A_E = \{ (u, v) | y_{uv} = 1\}$ the set of active edges and
$A_V = S \cup T$, where $S=\{ u | s_u = 1\}$ and  $T = \{v \text{ reachable from some } u\in S \text{ via a path of only active edges}\}$. We need to show that 
$
|A_V| = |S| + |A_E|
$
or equivalently we need to show
\[
|A_E| = |T|.
\]
The constraints (19) guarantee that each node $v \in V$ has at most one incoming active edge. The constraints (20) forbid cycles to form among the active edges, thus each active edges will point to exactly one active node that is not in $S$ (otherwise we would have a cycle). Thus, we have an one-to-one mapping between the active edges and the active nodes that are not in $S$, i.e., $|A_E| = |T|$.
\end{proof}
\subsection{Sample Average Approximation E-SAA}
Similarly, we construct a Sample Average Approximation, called E-SAA. Given $T$ graph samples, $G^1,G^2,\cdots,G^T$ from $\mathcal D$, we have
\begin{align}
\label{esaa:obj} \max & \frac{1}{T}\sum_{l=1}^T \displaystyle\sum_{ (u, v)\in E} b(v) y_{uv}^l +\sum_{u\in V} b(u)s_u\\       
     \mbox{s. t.\,} & \sum_{v\in V} s_v c_v  \leq \kappa  \\
 %&\sum_{w \in N^-(u)} y_{wu}^l+ s_{u} \geq y_{uv}^l,\   (u, v) \in E^l, l=1..T\\
 %\label{esaa:tree}&\sum_{u \in N^-(v)} y_{uv}^l+s_v \leq 1 ,\  v\in V, l=1..T \\
%\label{esaa:cycle}&\sum_{(u,v) \in C} y_{uv}^l \leq |C|-1 ,\ \text{ any cycle } C\\
%&s_u \in \{0,1\},\  u \in V, l=1..T\\
%&\label{esaa:range} y_{uv}^l \in [0, 1],\  (u, v) \in E, l=1..T
    &\text{Constraints (\ref{ipcycle:seed})-(\ref{ipcycle:range}) for } G^l, l=1..T
\end{align}
%where $y_{uv}^l \in [0, 1]$ is an abbreviation for $y_{uv}^{{G^l}}$, the indicator of edge $(u, v)$ in sample graph $G^l$.

We shall refer to the above mixed integer linear programming as E-SAA.

\subsubsection{Sample complexity and Time complexity of E-SAA}

Similar to that in Lemma~\ref{lem:texactsamp}, we can obtain the same bound on the number of samples.
\begin{lemma}
	For fixed $T= \Omega(\frac{1}{\epsilon^2} \log \frac{2}{\delta} \log |\mathcal C|)$,  we have     
    \[
    \Pr[|\bar B_E(S) - \B(S)| > \epsilon \sigma_B] \leq \delta,
    \]
%  where $\sigma_B$ is the total benefit of all nodes, defined in Eq.~\ref{eq:sumb}.
\end{lemma}

Since the number of sub-tour elimination constraints in Eq.~(\ref{ipcycle:cycle}) can be exponentially many, in theory the worst-case complexity of E-SAA is much worse than that in T-SAA. However, in practice, those constraints are added gradually and potentially lead to a more overall efficient formulation.

\vspace{-0.1in}
\subsection{Scalability Issues of Traditional Stochastic Optimization}
While both T-EXACT and E-EXACT (called EXACT for short) are designed based on a standard method for stochastic programming, traditional methods can only be applied for small networks, up to few hundreds nodes \cite{Shapiro14}.

There are three major scalability issues when applying SAA and using EXACT for the influence maximization problem.  First, the samples have a large size $O(m)$. For large networks, $m$ could be of size million or billion. As a consequence,  we can only have a small number of samples, sacrificing the solution quality. For billion scale networks, even one sample will let to an extremely large ILP, that exceeds the capability of the best solvers. Second, the theory developed to assess the solution quality such as those in \cite{Bayraksan06, Shapiro14}, only provide approximate confidence interval. That is the quality assessment is only justified for sufficiently large samples and may not hold for small sample sizes. And third, most existing solution quality assessment methods \cite{Bayraksan06, Shapiro14} only provide the assessment for a given number of sample size. Thus, if the quality requirement is given a priori, e.g., $(\epsilon, \delta)$ approximation, there is not an efficient algorithmic framework to identify the number of necessary samples.

\vspace{-0.1in}
\section{\ISSA{} - An Efficient $(1-\epsilon)$-optimal Solution}\label{sec:tiptop}

In this section, we introduce our main contribution \ISSA{}, which is the first algorithm that can return a $(1 - \epsilon)$-approximation ratio w.h.p of $(1 - \delta)$ where $\delta$ is given a priori. It overcomes the above mentioned scalability issues and can run on billion-scale networks.

\subsection{\ISSA{} Algorithm Overview}
For readability, we start with the solution to \CTVM{} in which all nodes have uniform cost. Let $\kappa = k$, we want to find $S$ with $|S| \le k$ so as to maximize the benefit $\mathbb{B}(S)$. Note that benefit function is still heterogeneous. The solution to non-uniform cost is presented later in Subsection \ref{sub:ctvm}.

At a high level, \ISSA{} first generates a collection $\mathcal{R}$ of random hyperedges sets which serves as the searching space to find a candidate solution $\hat S_k$ to \CTVM{}. It next calls the \Verify{} procedure which independently generates another collection of random hyperedges sets to closely estimate the objective function's value of the candidate solution. If this value is not close enough to the optimal solution, \ISSA{} generates more samples by calling the \AddSamples{} procedure to enlarge the search space, and thus finding another better candidate solution. When the objective function's value of the candidate solution is close enough to the optimal one, \ISSA{} halts and returns the found solution. The pseudo-code of \ISSA{} is presented in Alg. \ref{algo:issa}.

%\label{sec:algorithm}
\begin{algorithm}\small
	\caption{\ISSA{} Algorithm}
	\label{algo:issa}
	\begin{algorithmic}
		\State {\textbf{Input}: Graph $G=(V, E, b, c, w)$, seed set size $k >0$, and $\epsilon, \delta \in (0, 1)$.}
		\State {\textbf{Output}: Seed set $S_k$.}
        \State {1: \text{ } $\Lambda \leftarrow (1 + \epsilon)(2+\frac{2}{3}\epsilon)\frac{1}{\epsilon^2} \ln \frac{2}{\delta}$ }
        \State {2: \text{ } $t \leftarrow 1; t_{max} = \lceil 2\ln n/\epsilon \rceil; v_{max} \leftarrow 6$ }
        \State {3: \text{ } $\Lambda_{max} \leftarrow (1 + \epsilon)(2+\frac{2}{3}\epsilon)\frac{2}{\epsilon^2}(\ln \frac{2}{\delta/4}+ \ln {n \choose k})$}
		\State {4: \text{ } \text{Generate random $\Lambda$ hyperedges sets $R_1, R_2, \ldots$ using BSA \cite{Nguyen161}}}
		\State {5: \text{ }\textbf{repeat}}
        \State {6: \text{ }\text{ \ \ } $N_t \leftarrow \Lambda \times e^{\epsilon t}; \mathcal R_t \leftarrow \{R_1, R_2,\ldots,R_{N_t} \}$ }
   		\State {7: \text{ }\text{ \ \ } $\hat S_k\leftarrow \text{ILP}_{MC}(\mathcal R_t, c, k)$}	        
        \State {8: \text{ }\text{ \ \ } $<passed, \epsilon_1> \leftarrow  $\Verify$(\hat S_k,v_{max}, \epsilon, t_{max}, 2^{v_{max}} N_t)$}       
        \State {9: \text{ }\text{ \ \ } \textbf{if} (not $passed$) and ($Cov_{\R}(\hat S_k) \leq \Lambda_{max}$) \textbf{then}\\ 
        \text{ }\text{ \ \ \  \ \ \ \ \  \ } $t \leftarrow$ \AddSamples$(t, \epsilon, \epsilon_1)$}
        \State {10: \text{ }\textbf{until} $passed$ or $Cov_{\R}(\hat S_k) > \Lambda_{max}$}
		\State {11: \textbf{return} $\hat S_k$}
	\end{algorithmic}
\end{algorithm}

As \CTVM{} considers arbitrary benefits, we utilize Benefit Sampling Algorithm -- \BRIS{} in \cite{Nguyen161} to embed the benefit of each node into consideration, shown in line 4 of Algorithm \ref{algo:issa}.  \BRIS{}  performs a reversed influence sampling (\textsf{RIS}) in which \emph{the probability of a node chosen as the source is proportional to its benefit}. Random hyperedges generated via \BRIS{} can capture the ``benefit landscape''. That is they can be used to estimate the benefit of any seed set $S$ as stated in the following lemma.

\begin{lemma}\cite{Nguyen161}
	\label{lem:ris}
	Given a fixed seed set $S \subseteq V$, and let $R_1, R_2,\ldots, R_j,\ldots$ be random hyperedges sets generated using Benefit Sampling Algorithm \cite{Nguyen161}, define random variables
    \begin{align}
    	\label{def:randvar}
		Z_j = \left\{ \begin{array}{ll}
			1 & \text{ if } R_j \cap S \neq \emptyset, \\
			0 & \text{ otherwise}.
		\end{array}
		\right.
	\end{align}
then
	\begin{align}
	\E[Z_j] = \Pr[ R_j \cap S \neq \emptyset ] = \frac{\B(S)}{\Gamma}
	\end{align}
where $\Gamma = \sum_{v \in V} b(v)$ is the total nodes' benefit.
\end{lemma}

% \BRIS{} has two important advantages. First, it provides a simple $0$-$1$ Bernoulli random variable to estimate the benefit of any seed set $S$. This allows the use of probability theory to effectively bound the number of necessary samples.  Secondly, \BRIS{} integrate the nodes' benefit into the \textsf{RIS} samples, i.e., all the generated samples are equally important from the optimization perspective.
 For a collection of $T$ random hyperedges sets $\R=\{R_1, R_2,\ldots,R_T\}$, we denote by 
 \[
 Cov_{\R}(S) = \sum_{j=1}^T {Z_j}, 
 \] 
 the number of hyperedges sets that intersect $S$, and 
 \[
 \B_\R( S ) = \frac{Cov_{\R}(S)}{T} \times \Gamma,
 \]
 the estimation of $\B(S)$ via $\R$. 
 
As shown in Theorem \ref{thm:main}, \ISSA{} has an approximation of $(1 - \epsilon)$ with a probability of at least $(1 - \delta)$. The key point to improve the current best ratio of $(1 - 1/e - \epsilon)$ (and $(1 - 1/\sqrt{e} - \epsilon)$) to $(1 - \epsilon)$ lies in solving the Maximum Coverage (MC) problem after generating $\mathcal{R}$ random hyperedges sets of samples (line 7 of Alg. \ref{algo:issa}). Instead of using greedy technique as in all existing algorithms, we solve the MC {\em exactly} using Integer Linear Program (ILP) (detailed in subsection \ref{sec:ILP}).

As exactly solving ILP for MC is NP-hard itself, we need to reduce the time and memory complexities as much as possible. This becomes the solely drive force for the design of our algorithm, which is handled as follows.

First, we need to keep the ILP search space small at the first phase during the searching for the candidate solution $\hat S_k$. Else, the ILP solver cannot be executed. It is worth noting that existing solutions only focused on reducing the {\em total} number of samples generated, not at the searching phase. Relevant to our approach, SSA \cite{Nguyen163} does have the first phase, however, it has fixed parameter setting, thus cannot achieve optimal number of samples for this phase. This enforces us to carefully generate the first set $\mathcal{R}$ with size $\Lambda$ as shown in line 1 of Alg. \ref{algo:issa}. And once $\hat S_k$ is not good enough, \AddSamples{} will generate an additional set of samples, which should be dynamically determined in order to meet the requirement of ILP, that is, it should be just large enough to find a near-optimal candidate solution to \CTVM{}. We will discuss more details about \AddSamples{} later in subsection \ref{sec:proc}.

Second, in an effort to keep the ILP size small, we need to avoid executing \AddSamples{} as much as possible. Therefore, we need to put more effort in proving the quality of candidate solutions, which is handled by \Verify{}, described in subsection \ref{sec:proc}.

%Phase 1: generated sample to find solution (limit sample in the first phase)
%Phase 2: To make sure the sample is enough. (put more effort in proving the quality of solution), much smaller $\epsilon_2$.  (more efficient probing method)

%SSA (limitation: fix parameter setting, cannot achieve optimal of samples), balance for searching for optimal solution)
%D-SSA (limitation: number of samples for checking is as many as number of searching of solution).

%Constraint on the number of samples for checking... No tradeoff between checking and searching. (Searching small and checking can be very big)
%\subsection{Benefit Sampling}
%\label{sub:bsa}

 We discuss the rest of \ISSA{} in the following subsections.

\vspace{-0.1in}
\subsection{Mixed Integer Linear Programming $ILP_{MC}$}\label{sec:ILP}
Given a collection of hyperedges sets $\mathcal R$ and the cost $c(v)$ of selecting nodes $v \in V$ and a seek size $k$, we formulate the following ILP$_{MC}(\mathcal R, c, k)$, to find the optimal solution over the generated hyperedges sets $\mathcal R$ to the MC problem.
\begin{small}
 \begin{align}
    \text{maximize} & \sum_{R_j \in \mathcal R} (1-y_j) &\\
    \text{subject to} &\sum_{v\in V}{s_v} \leq k & \label{eq:17}\\
        &\sum_{v \in R_j} s_v + y_j \geq 1 & \forall R_j \in \mathcal R\label{eqn:sample-constraint}\\
        &s_i \in \{0,1\} & y_j \in [0, 1]
 \end{align}
 \end{small}
Here, $s_v=1$ iff node $v$ is selected into the seed set, and $s_v =0$, otherwise. The variable $y_j=1$ indicates that the hyperedges sets $R_j$ cannot be covered by the seed set ($S=\{ v | s_v = 1\}$) and $y_j=0$, otherwise. The objective aims to cover as many hyperedges sets as possible while keeping the cost at most $k$ using the constraint (\ref{eq:17}).

We note that the benefit in selecting the node $b(u)$ does not appear in the above ILP as it is embedded in the Benefit Sampling Algorithm (BSA) in \cite{Nguyen161}.

On one hand, the above ILP can be seen as a Sample Average Approximation (SAA) of the CTVM problem as it attempts to find optimal solutions over the randomly generated samples. On the other hand, it is different from the traditional SAA discussed in Sec. \ref{sec:stoc} as it does not require the realization of all random variables, i.e., the status of the edges. Instead, only a local portion of the graph surrounding the sources of hyperedges sets need to be revealed. This critical difference from traditional SAA significantly reduces the size of each sample, effectively, results in a much more compact ILP. %Also, except for the carnality constraint, the rest are simple covering constraints of which the presolving and cutting-planes techniques are matured in commercial solver like CPLEX \cite{}.

\subsection{The \Verify{} and \AddSamples{} Procedures}\label{sec:proc}

As shown in Alg. \ref{algo:verify}, \Verify{} takes a candidate solution $\hat S_k$, precision limit $v_{max}$, and the maximum number of hyperedges sets $T_{cap}$ as the input.  It keeps generating hyperedges sets to estimate  $\B(\hat S_k)$ until either the relative error reaches $\epsilon/2^{v_{max}-1}$ or the maximum number of generated samples $T_{cap}$ is reached. 

 \Verify{} uses the stopping rule algorithm in \cite{Nguyen163} to estimate the influence. It generates a new pool of  random hyperedges sets, denoted by $\mathcal{R}_{ver}$. For each pair $\epsilon_2', \delta_2'$, derived from $\epsilon_2$ and $\delta_2$, the stopping rule algorithm will stop when either the cap $T_{max}$ is reached or there is enough evidence (i.e. $cov \geq \Lambda_2$) to conclude  that
 \[
 Pr[(1-\epsilon_2')\B(\hat S_k) \le  \B_{\R_{ver}} (\hat S_k) \leq (1 + \epsilon_2') \B(\hat S_k)] \ge 1 - \delta_2'.
 \]

The value of $\delta_2'$ is selected as in line 3 so that the probability of the union of all the bad events is bounded by $\delta_2$.

If the stopping rule algorithm stops within $T_{cap}$ hyperedges sets, the algorithm evaluates the relative difference $\epsilon_1$ between the estimations of $\B(\hat S_k)$ via $\R_{t}$ and $\R_{ver}$. It also estimates 
the relative gap $\epsilon_3$ between $\B(S^*_k)$ and its estimation using $\R_{t}$. If the combined gap $(1-\epsilon_1)(1-\epsilon_2)(1-\epsilon_3) < (1-\epsilon)$, \Verify{} returns `true' and goes back to \ISSA{}. In turn, \ISSA{} will return $\hat S_k$ as the solution and terminate. 

If $\hat S_k$ does not pass the check in \Verify{}, \ISSA{} uses the sub-procedure \AddSamples{} (Alg. \ref{algo:add}) to increase the size of the hyperedges sets.  Having more samples will likely lead to better candidate solution $\hat S_k$, however, also increase the ILP solving time. Instead of doubling the current set $\mathcal{R}$ as SSA does, we carefully use the information in the values of $\epsilon_1$ from the previous round together with $\epsilon$ to determine the increase in the sample sizes. Recall that the sample size is $e^{t\epsilon} \Lambda$ for increasing integer $t$. Thus, we increase the sample size via increasing $t$  by (approximately) $\log_{e^\epsilon} \frac{\epsilon_1^2}{\epsilon^2}$. We force $t$ to increase by at least one and at most $\Delta t_{max} = \lceil 2/\epsilon \rceil$. That is the number of samples will increase by a multiplicative factor between $e^{\epsilon} \approx (1+\epsilon)$ and $e^{\Delta t_{max}} \approx e^2$.

\begin{algorithm}\small
	\caption{\Verify}
	\label{algo:verify}
	\begin{algorithmic}
		\State {\textbf{Input}: Candidate solution $\hat S_k$,  $v_{max}$, $\epsilon$, $t_{max}$, and $T_{cap}$.}
		\State {\textbf{Output}: Passed/not passed and $\epsilon_1$.}
        \State { 1: \text{ } $\mathcal R_{ver}\leftarrow \emptyset,  \delta_2 = \frac{\delta}{4}, cov=0,  \epsilon_1=\epsilon_2=\infty$ }
        \State { 2: \text{ }\textbf{for} $i\leftarrow 0$ to $v_{max} -1$ \textbf{do}}
   		\State { 3: \text{ }\text{ \ \ } $\epsilon_2 = \min\{\epsilon,1\}/2^i, \epsilon_2' = \frac{\epsilon_2}{1-\epsilon_2};\delta_2' = \delta_2/(v_{max} \times t_{max} ) $}
        \State { 4: \text{ }\text{ \ \ } $\Lambda_2 = 1+(2+2/3\epsilon_2')(1 + \epsilon_2')\ln \frac{2}{\delta_2'} \frac{1}{(\epsilon_2')^2}$}	    
		\State { 5: \text{ }\text{  \  } \textbf{while} $cov < \Lambda_2$ \textbf{do}}        
   		\State { 6: \text{ }\text{  \ \ \ } Generate $R_j$ with BSA \cite{Nguyen161} and add it to $\mathcal R_{ver}$}
%		\State { 7: \text{ }\text{  \ \ \ } $T \leftarrow T + 1$}        
		\State { 8: \text{ }\text{  \ \ \ } \textbf{if} $R_j \cap S \neq \emptyset$ \textbf{then} $cov = cov + 1$}
        \State { 9: \text{ }\text{  \ \ \ } \textbf{if} $|\mathcal R_{ver}| > T_{cap}$ \textbf{then} return $<false, \epsilon_1>$}
        \State {10: \text{ }\text{  \ } \textbf{end while} }
		\State {11: \text{ }\text{  \ } $\B_{ver}(\hat S_k) \leftarrow  \Gamma 
        \frac{cov}{|\mathcal R_{ver}|}$, $\epsilon_1 \leftarrow  1-\frac{\B_{ver}(\hat S_k)}{\B_{\mathcal R}(\hat S_k)}$ }                
		\State {12: \text{ }\text{  \ } \textbf{if} $(\epsilon_1 > \epsilon)$ \textbf{then} return $<false, \epsilon_1>$}        
		\State {13: \text{ }\text{  \ }  $\epsilon_3 \leftarrow 
 \sqrt{\frac{3\ln ( t_{max}/\delta_1)}{(1-\epsilon_1)(1-\epsilon_2)Cov_{\R_t}(\hat S_k) } }$.}
		\State {14: \text{ }\text{  \ } \textbf{if} $(1-\epsilon_1)(1-\epsilon_2)(1-\epsilon_3)>(1-\epsilon)$ \textbf{then}\\
		\text{ }\text{ \ \ \  \ \ \ \ \  \ } return $<true, \epsilon_1>$}
        \State {15: \text{ }\textbf{end for}}  
		\State {16: return $<false, \epsilon_1>$}
	\end{algorithmic}
\end{algorithm}

\begin{algorithm}\small
	\caption{\AddSamples}
	\label{algo:add}
	\begin{algorithmic}
		\State {\textbf{Input}:  $t$ and $\epsilon_1$.}
		\State {\textbf{Output}: $t$.}
        \State {1: \text{ } $\Delta t_{max} = \lceil 2/\epsilon \rceil$ }
        \State {2: \text{ } return $t + \min\{ \max \{ \lceil 1/\epsilon \ln \frac{\epsilon_1^2}{\epsilon^2}\rceil, 1\}, \Delta t_{max} \}$ }
	\end{algorithmic}
\end{algorithm}

\vspace{-0.1in}
\section{Optimality of \ISSA{}}\label{sec:complexity}
In this section, we prove that  \ISSA{} for arbitrary cost \CTVM{} problem returns a solution $\hat S$ that is optimal up to a multiplicative error $1-\epsilon$ with high probability. Fig. \ref{fig:ssa} shows the proof map of our main Theorem \ref{thm:main}. We first prove the optimality of \ISSA{} in the case of the uniform cost and then extend it to arbitrary cost in subsection \ref{sub:ctvm}.

\vspace{-0.1in}
\subsection{Uniform Cost \CTVM{}}\label{sec:uniform}
 
\begin{figure}[!ht]
%	\vspace{-0.15in}
	\centering
	\includegraphics[width=0.7\linewidth]{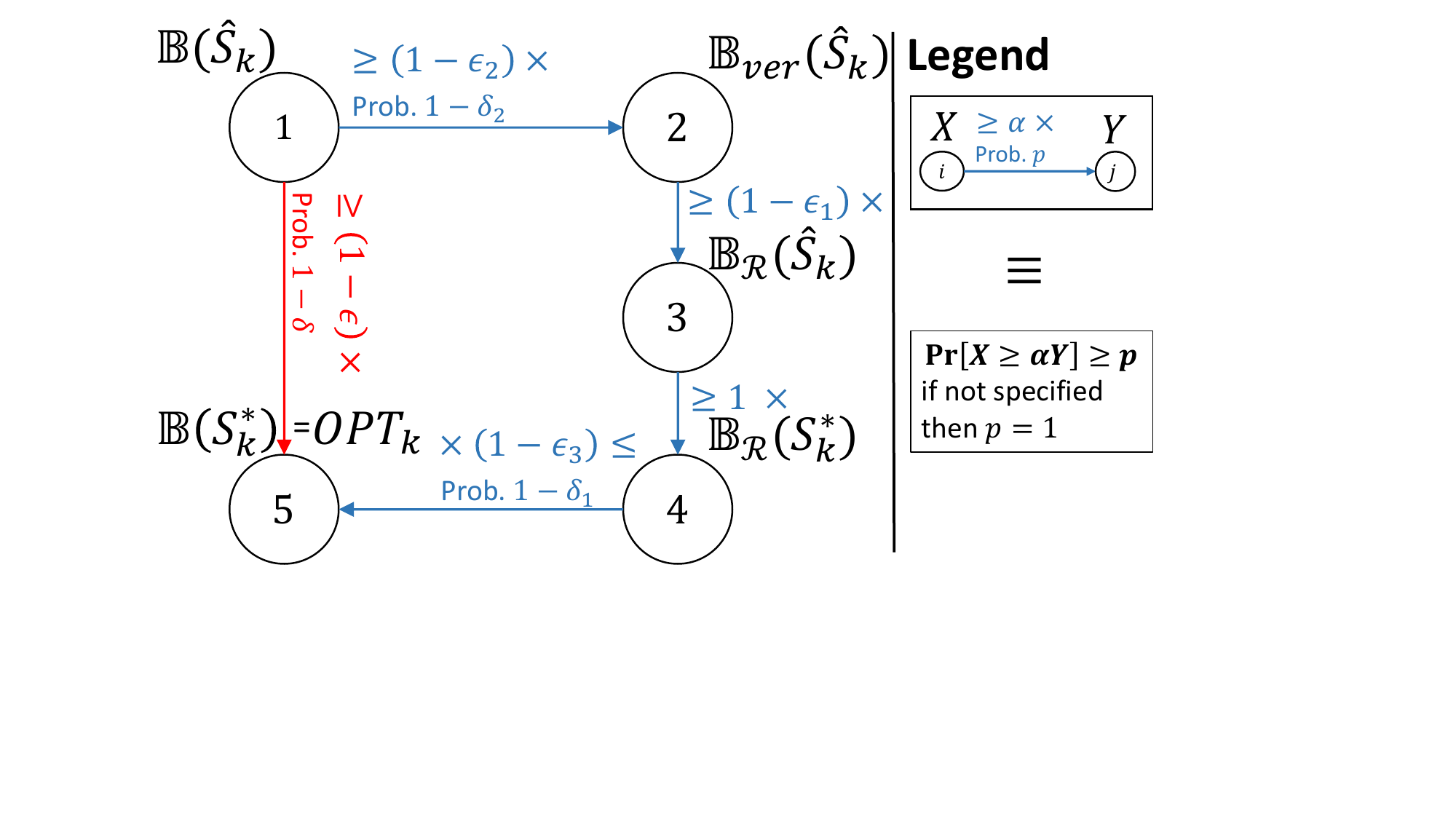}
	\caption{Proof map of the main Theorem \ref{thm:main}}
	\vspace{-0.08in}
	\label{fig:ssa}
	\vspace{-0.1in}
\end{figure}

Let $R_1, R_2, R_3,\ldots, R_j,\ldots$ be the random hyperedges sets generated in \ISSA{}. Given a seed set $S$, define random variables $Z_j$ as in (\ref{def:randvar}) and $Y_j = Z_j - \E[Z_j]$. Then $Y_j$ satisfies the conditions of a \emph{martingale} \cite{Chung06}, i.e., $\E[Y_i | Y_1, Y_2,\ldots,Y_{i-1}] = Y_{i-1}$ and $\E[Y_i]<+\infty$. This martingale view  is adopted from \cite{Tang15} to cope with the fact that random hyperedges sets might not be independent due to the stopping condition: the later hyperedges sets are generated only when the previous ones do not satisfy the stopping conditions. We obtain the same results in Corollaries 1 and 2 in \cite{Tang15}.
	\begin{lemma}[\cite{Tang15}]
	\label{lem:chernoff}
	Given a set of nodes $S$ and random hyperedges sets $\R = \{R_j\}$ generated in \ISSA,  define random variables $Z_j$  as in (\ref{def:randvar}). Let $\mu_Z=\frac{\B(S)}{\Gamma}$ and $\hat \mu_Z = \frac{1}{T}\sum_{i=1}^{T}Z_i$ be an estimation of $\mu_Z$, for  fixed $T > 0$. For any $0 \leq \epsilon$, the following inequalities hold
	\begin{align}
		\Pr[\hat \mu \geq  (1+ \epsilon) \mu]  &\leq
		e^{\frac{-T\mu\epsilon^2}{2 + \frac{2}{3}\epsilon}}, \text{ and}\\
				\Pr[\hat \mu \geq  (1+ \epsilon) \mu]  &\leq
	\label{eq:chernoff3}	e^{\frac{-T\mu\epsilon^2}{3}}, \text{ and}\\
		\Pr[\hat \mu \leq  (1- \epsilon) \mu] &\leq e^{\frac{-T\mu\epsilon^2}{2}}.
	\end{align}
	\label{lem:error_bound}
 \end{lemma}
A common framework in \cite{Tang14, Tang15,Nguyen163} is to generate random hyperedges sets and  use the greedy algorithm to select $k$ seed nodes that cover most of the generated hyperedges sets. It is shown in Lemma 3 \cite{Tang14} that  $(1-1/e-\epsilon)$ approximation algorithm with probability $1-\delta$ is obtained when the number of hyperedges sets reaches a threshold
\begin{align}
\label{eq:theta}
	\theta(\epsilon, \delta) = c\times (8+\epsilon) \left(\ln \binom{n}{k} + \ln\frac{2}{\delta}\right) \frac{n}{OPT_k}\frac{1}{\epsilon^2},
\end{align}
for some constant $c>0$.

The constant $c$ is bounded to be $8+\epsilon$ in \cite{Tang14}, brought down to $8(e-2)(1-1/(2e))^2 \approx 3.7$ in \cite{Nguyen161} using the zero-one estimator in the work of Dagum et al. \cite{Dagum00}. And the current best is $c=2+2/3 \epsilon$, inducted from Lemma 6 in \cite{Tang15}.
\begin{comment}
\[
c=2+2/3 \epsilon,
\]
inducted from Lemma 6 in \cite{Tang15}.
\end{comment}

Note that $\theta(\epsilon, \delta)$ cannot be used to decide how many hyperedges sets we need to generate since $\theta$ depends on the unknown value $OPT_k$, of which computation is \#P-hard. 

To overcome that hurdle, a simple stopping rule is developed in \cite{Nguyen161} to check on whether we have sufficient hyperedges sets to guarantee, w.h.p., a $(1- \epsilon)$ approximation. The rule is that we can stop when we can find \emph{any} seed set $S_k$, of size $k$, that coverage $Cov_{\mathcal R}(S_k)$ exceeds 
\begin{align}
\Lambda_{max} &= (1+\epsilon) \theta(\epsilon, \delta/4) \times \frac{OPT_k}{n}\\
		&= (1+\epsilon) (2+\frac{2}{3}\epsilon)\frac{1}{\epsilon^2}(\ln \frac{2}{\delta/4}+ \ln {n \choose k})
\end{align}

Our algorithm \ISSA{} utilizes this stopping condition to guarantee that at most $O(\theta(\epsilon, \delta))$ hyperedges sets are used in the ILP in the worst-case. As a result, the size of the ILP is kept to be almost linear size, assuming $k \ll n$. 
\begin{theorem}
\label{tiptop:size}
The expected number of non-zeros in the ILP of \ISSA{} is 
\[
O\left( \left(\ln {n \choose k} + \ln\frac{2}{\delta}\right) \frac{n}{\epsilon^2}\right).
\]
\end{theorem}
\begin{proof}
	The expected number of hyperedges sets is $O((1+\epsilon)\theta(\epsilon, \delta))$. Moreover, the expected size of each hyperedges sets is upper-bounded by $\frac{OPT_k}{k}$ (Lemma 4 and Eq. (7) in \cite{Tang14}). Thus, the size of the ILP which is equal the total sizes of all the hyperedges sets plus $n$, the size of the cardinality constraint, is at most $O((1+\epsilon)\theta(\epsilon, \delta)  OPT_k) = O\left( \left(\ln {n \choose k} + \ln\frac{2}{\delta}\right) \frac{n}{\epsilon^2}\right)$
\end{proof}

In practice, the number of required samples in our ILP is many times larger than the minimum number of required samples, i.e., the ILP size is much smaller than the worst-case bound in the above theorem, as shown in Section \ref{sec:exp}.

\begin{lemma}
\label{lem:opt_bound}
Let $S^*_k$ be a seed set of size $k$ with maximum benefit, i.e., $\B(S^*_k) = OPT_k$. Denote by $\mu^*_k=\frac{OPT_k}{\Gamma}, \delta_1 = \delta/4$ and 
\[
\epsilon^*_{t} = \sqrt{\frac{3\ln (t_{max}/\delta_1)}{N_t \mu^*_k} }.
\]
We have: 
\[
\Pr[ \B_{\R_t}(S^*_k) \geq (1-\epsilon^*_t) \B(S^*_k)\  \mathbf{\forall t=1..t_{max}}] \geq 1 - \delta_1.
\]
\end{lemma}
\begin{proof}
Apply  Lem. \ref{lem:chernoff} (Eq. \ref{eq:chernoff3}) for seed set $S^*_k$, mean $\mu^*_k$, $\epsilon^*_t$ and $N_t$ samples. For each $t\in [1..t_{max}]$, we have 
%\todo[inline]{Can extend the actual algebra here for the journal}
\begin{align}
&\Pr[ \B_{\R_t}(S^*_k) < (1-\epsilon_t) \B(S^*_k) ] \\
&=\Pr\left[ \frac{\B_{\R_t}(S^*_k)}{\Gamma} < (1-\epsilon_t) \frac{\B(S^*_k)}{\Gamma} \right]\\
&< e^{-\frac{3N_t\mu^*_k\epsilon^{*2}_{t}}{3}} = e^{-\ln(\frac{t_{max}}{\delta_1})}< \delta_1/t_{max}.
\end{align}

Taking the union bound over all $t\in [1, t_{max}]$ yields the proof.
\end{proof}

\begin{lemma}\label{lem:5}[\circled{1}$\rightarrow$\circled{2}]
The probability of the bad event that there exists  some set of $t \in [1, t_{max}], i \in [0, v_{max}-1]$, and $\epsilon_2 = \epsilon/2^i$ so that the \Verify{} algorithm returns some bad estimation of $\hat S_k$ at line 10, i.e.,
\begin{align*}
\B_{ver}(\hat S_k) > \frac{1}{(1-\epsilon_2)}\B(\hat S_k),
\end{align*}
is less than $\delta_2$. Here $\B_{ver} (\hat S_k) = \Gamma\times \frac{Cov_{\mathcal R_{ver}}(\hat S_k)}{|\mathcal R_{ver}|}$ be an estimation of $\B(\hat S_k)$ using random hyperedges sets in $\mathcal R_{ver}$.
\end{lemma}
\begin{proof}
After reaching line 10 in \Verify{}, the generated hyperedges sets within \Verify{}, denoted by $\mathcal R_{ver}$,  will satisfy the condition that
\vspace{-.1in}
\[
cov = Cov_{\mathcal R_{ver}}(\hat S_k) \geq \Lambda_2 = 1+(2+\frac{2}{3}\epsilon_2')(1 + \epsilon_2')\ln \frac{2}{\delta_2'} \frac{1}{\epsilon_2'^2},
\]
where $\epsilon_2' = \frac{\epsilon_2}{1-\epsilon_2}$.

According to the stopping rule theorem\footnote{replacing the constant $4(e-2)$ with $2+2/3\epsilon$ due to the better Chernoff-bound in Lemma \ref{lem:chernoff}} in \cite{Dagum00}, this stopping condition guarantees that
\begin{align*}
\Pr[   \B_{ver}(\hat S_k) > (1+\epsilon_2')\B(\hat S_k) ] < \delta_2'.
\end{align*}
Substitute $\epsilon'_2=\frac{\epsilon_2}{1-\epsilon_2}$ and simplify, we obtain
\begin{align*}
\Pr[ &\B_{ver}(\hat S_k) > \frac{1}{1-\epsilon_2}\B(\hat S_k)] < \delta_2'.
\end{align*}

The number of times \Verify{} invoked the stopping condition to estimate $\hat S_k$ is at most $t_{max} \times v_{max}$. Thus, we can use the union bound of all possible bad events to get a lower-bound $\delta_2' \times t_{max} \times v_{max} = \delta_2$ for the probability of existing a bad estimation of $\hat S_k$.
\end{proof}
\circled{2}$\rightarrow$\circled{3}: This holds due to the definition of $\epsilon_1$ in \Verify{}. Since $\epsilon_1 \leftarrow  1-\B_{ver}(\hat S_k)/\B_{\mathcal R_t}(\hat S_k)$, it follows that 
\[
\B_{ver}(\hat S_k)=  (1-\epsilon_1)\B_{\mathcal R_t}(\hat S_k).
\]
We note that it is possible that $\epsilon_1 < 0$, and the whole proof still goes through even for that case. In the experiments, we, however, do not observe negative values of $\epsilon_1$.

\circled{3}$\rightarrow$\circled{4}: Since $\hat S_k$ is an optimal solution of ILP$_{MC}(\mathcal R, c, k)$, it will be the $k$-size seed set that intersects with the maximum number of hyperedges sets in $\mathcal R$. Let $S^*_k$ be an optimal $k$-size seed set, i.e., the one that results in the maximum expected benefit\footnote{If there are multiple optimal solutions, we break the tie by using alphabetical order on nodes' ids.}. Since $|S^*_k| = k$, it follows that 
\[
\B_{\mathcal R_t}(\hat S_k) \geq \B_{\mathcal R_t}(S^*_k),
\]
where $\B_{\mathcal R_t}(\hat S_k)$ and  $\B_{\mathcal R_t}(S^*_k)$ denote the number of hyperedges sets in $\mathcal R_t$ that  intersect with $\hat S_k$ and $S^*_k$, respectively.

\begin{theorem}\label{thm:main}[Main theorem \circled{1}$\rightarrow$\circled{5}]
Let $\hat S_k$ be the solution returned by \ISSA{} (Algorithm \ref{algo:issa}). We have 
\begin{align}
\Pr[\B(\hat S_k) \geq (1-\epsilon) OPT_k] \geq 1 - \delta
\end{align}
\end{theorem}
\begin{proof}
First we use the union bound to bound the probability of the bad events. Then we show if none of the bad events happen, the algorithm will return a solution satisfying $\B(\hat S_k) \geq (1-\epsilon) OPT_k$. The bad events and the bounds on their probabilities are
\begin{enumerate}
\item $\Pr[ \exists t: \B_{\R_t}(S^*_k) < (1-\epsilon^*_t) \B(S^*_k)\  ] < \delta_1$ (Lem. \ref{lem:opt_bound})
\item $\Pr[ \exists t, i, \epsilon_2 = \frac{\epsilon}{2^i}: \B_{\R_{ver}}(\hat S_k) > \frac{1}{1-\epsilon_2}\B(\hat S_k)] < \delta_2$ (Lem. \ref{lem:5})
\item $\Pr[ (Cov_{\R_{t}}(\hat S_k) > \Lambda_{max}) \\ \text{ and } (\B(\hat S_k) < (1-\epsilon) OPT_k) ] < \delta/4$
\item $\Pr[ (t > t_{max}) \text{ and } (Cov_{\R_{t}}(\hat S_k) \leq \Lambda_{max})] < \delta/4$
\end{enumerate}
The bounds in 1) and 2) come directly from Lems. \ref{lem:opt_bound} and \ref{lem:5}. The bound in 3) is a direct consequence of the stopping condition algorithm in \cite{Nguyen161}. The bound in 4) can be 
shown by noticing that when $t>t_{max} $ then $N_t \gg \theta(\epsilon, \delta)$. Apply the union bound, the probability that none of the above bad events happen is, hence, at most $\delta_1+\delta_2+\delta/4 + \delta/4 = \delta$. 

\textbf{Assume that none of the above bad events happen}, we show that \ISSA{} returns a $(1-\epsilon)$ optimal solution.
If \ISSA{} stops with $Cov_{\R_{t}}(\hat S_k) > \Lambda_{max}$, it is obvious that $\B(\hat S_k) \geq (1-\epsilon)OPT_k$, since the bad event in 3) do not happen. Otherwise, algorithm \Verify{}  returns `true' at line 13 for some $t \in [1, t_{max}]$ and $i\in [0, v_{max})$.

Follow the path \circled{1}$\rightarrow$\circled{2}$\rightarrow$\circled{3}$\rightarrow$\circled{4}. No bad event in 2) implies
\begin{align}
\label{eq:hrt}
\noindent\B(\hat S_k) &\geq (1-\epsilon_2) \B_{ver}(\hat S_k) \geq  (1-\epsilon_2)(1-\epsilon_1)\B_{\R_t}(\hat S_k)\\
\label{eq:1234}\noindent &\geq (1-\epsilon_2)(1-\epsilon_1) \B_{\R_t}(S^*_k).
\end{align}

No bad event in 1) implies
$
	 \B_{\R_t}(S^*_k) \geq (1-\epsilon^*_t) \B(S^*_k).
$

\textbf{Claim}: [\circled{4}$\rightarrow$\circled{5}]
 \[ \B_{\R_t}(S^*_k) \geq (1-\epsilon_3) \B(S^*_k),\] 
 
To show the above inequality, we prove that 
\begin{align*}
  \epsilon^*_t &= \sqrt{\frac{3\ln (t_{max}/\delta_1)}{N_t \mu^*_k} }  \leq \epsilon_3 = \sqrt{\frac{3\ln ( t_{max}/\delta_1)}{(1-\epsilon_1)(1-\epsilon_2)Cov_{\R_t}(\hat S_k) } }\\
\Leftrightarrow & N_t \mu^*_k \geq (1-\epsilon_1)(1-\epsilon_2) Cov_{\R_t}(\hat S_k)\\
\Leftrightarrow & N_t \B(S^*_k)/\Gamma \geq (1-\epsilon_1)(1-\epsilon_2) N_t \B_{\R_t}(\hat S_k)/\Gamma\\
\Leftrightarrow & \B(S^*_k) \geq (1-\epsilon_1)(1-\epsilon_2) \B_{\R_t}(\hat S_k).
\end{align*}
The last one holds due to the optimality of $S^*_k$ and Eq. \ref{eq:hrt}.
\begin{align*}
\B(S^*_k) \geq  \B(\hat S_k) \geq (1-\epsilon_1)(1-\epsilon_2) \B_{\R_t}(\hat S_k).
\end{align*}

Combine Eq. \ref{eq:1234} and the above claim, we have
\begin{align*}
\B(\hat S_k) &\geq  (1-\epsilon_2)(1-\epsilon_1) \B_{\R_t}(S^*_k) \\
			 &\geq  (1-\epsilon_2)(1-\epsilon_1)(1-\epsilon_3) \B(S^*_k)\\
			 &\geq  (1-\epsilon) OPT_k.
\end{align*}
The last one holds due to the terminating condition $(1-\epsilon_2)(1-\epsilon_1)(1-\epsilon_3) < (1-\epsilon)$ in line 13 in the \Verify{} algorithm.
\end{proof}

%\todo[inline]{The expected number of RR sets generated by \ISSA{}  meets the Type-1 threshold in \cite{Nguyen163}}

\vspace{-0.1in}
\subsection{Arbitrary Cost \CTVM{}}
\label{sub:ctvm}
We now consider the case of heterogeneous cost. Note that \ISSA{} and its proofs in subsection \ref{sec:uniform} already considered the heterogeneous benefit function $b(.)$ thus we only discuss the arbitrary cost function $c(.)$ in this subsection.

{\bf Changes in the algorithm.} With heterogeneous selecting cost, seed sets may have different sizes. We define $k_{max} = \max\{k: \exists S \subset V, |S|=k, c(S) \leq \kappa \}$.  The value of $\Lambda_{max}$ at line 3 of Alg. \ref{algo:issa} will be defined as follows:
\[
\Lambda_{max} = (1 + \epsilon)(2+\frac{2}{3}\epsilon)\frac{1}{\epsilon^2}[ \ln (8/\delta) + \min\{ k_{max} \ln n, n\}].
\]
Line 7 of Alg. \ref{algo:issa} will be replaced by $\hat S_k\leftarrow \text{ILP}_{MC}(\mathcal R_t, c, \kappa)$ where we pass the value $\kappa$ instead of $k$.

In addition, the cardinality constraint (\ref{eq:17}) will be changed into a knapsack constraint
$
\sum_{v\in V}{c(v) s_v} \leq \kappa,
$
where $c(v)$ is the cost of selecting node $v$ and $\kappa$ is the given budget.

The Verify and IncreaseSamples procedures are kept intact.

\begin{theorem}The generalized \ISSA{} as discussed above has an approximation ratio of $(1 - \epsilon)$ with high probability.
\end{theorem}

\begin{proof}
We follow the same proof map as in subsection \ref{sec:uniform}. All the previous proofs of the convergence and correctness are still held. Note that we only use $\Lambda_{max}$ in Lemma \ref{lem:5}, in which we apply the inequality $\Lambda_{max} \leq 2n^{2+\epsilon}$. This inequality still holds with the new value of $\Lambda_{max}$.
\end{proof}
\subsubsection{Time complexity} Assume that we solve $\ISSA$ using exhaustive search that provides the same worst-case time complexity as other methods such as branch-and-cut. 

The number of variables $y_j$ in $\ISSA$ equals to the number of hyperedges and is bounded by  $O(n k \log n \frac{1}{\epsilon^2 OPT_k})$
from Eq.~(\ref{eq:theta}). Since $OPT_k \geq k$, the number of variables in $\ISSA$ is $
N_{TipTop}= O(n\log n \frac{1}{\epsilon^2})$. 

For each of the $2^n$ possible assignments of $s_v$, we need to solve a remaining linear programming of size 
\begin{align}
M_{TipTop}&=O\left( \left(\ln {n \choose k} + \ln\frac{2}{\delta}\right) \frac{n}{\epsilon^2}\right)
\\
&= O\left(n k \ln n \frac{1}{\epsilon^2}\right) = O(n^2 \log n 1/\epsilon^2).
\end{align}

Apply Theorem \ref{theo:karmarkar} and note the high concentration of the number of hyperedges and their total size around the above values, we have, 
\begin{lemma}
The expected time to solve $\ISSA$ will be
\[
O\left( 2^n n^{3.5} \frac{1}{\epsilon^7} n^4 \frac{1}{\epsilon^4} polylog(n) \right) = O\left(2^n n^{7.5} \frac{1}{\epsilon^{11} } polylog(n) \right).
\]
\end{lemma}
In theory, the worst-case of $\ISSA$ (with a multiplicative  $\epsilon$  error guarantee) is better than T-SAA (with an additive error $\epsilon \sigma_B$ guarantee) by a factor $n^{5.5}$. In practice, $\ISSA$ is also much more efficient due to its simple constraints.

\section{Experiments}\label{sec:exp}
We conduct several experiments to illustrate the performance and
utility of \ISSA{}. First, the performance of \ISSA{} is compared to
both T-SAA and E-SAA. Our results show that it is magnitudes
faster while maintaining solution quality (sec.
\ref{sec:ilp-runtime}). We then apply \ISSA{} as a benchmark for
existing methods, showing that they perform better than their guarantee in certain
cases---but have degraded performance in others. Finally,
we conclude the section with an in-depth analysis of \ISSA{}'s
performance, with a focus on its sampling behavior.

We implemented \ISSA{} in Rust using Gurobi to solve the IP.\footnote{Code available at \texttt{\url{https://github.com/emallson/tiptop}}} Unless
noted otherwise, each experiment is run 10 times and the results
averaged. Settings for $\epsilon$ are listed with each experiment,
but $\delta$ is fixed at $1/n$. All influence values are obtained by running a separate estimation program with $\epsilon = 0.02$ using the seed sets produced by each algorithm as input. Throughout this section, we scale solution quality by an \textit{upper bound} on the optimal solution. When comparing to T-SAA and E-SAA, we treat the maximal performance of these algorithms as optimal to show the performance of \ISSA{}. Subsequently, we assume that \ISSA{} exactly matches its approximation guarantee, which places an upper bound on the optimal of $1 / (1 - \epsilon)$ times the \ISSA{} solution. 

\begin{table}[t]\small
  \centering
  \begin{tabular}{l|lrr}
    \textbf{Dataset} & \textbf{Network Type} & \textbf{Nodes} & \textbf{Edges} \\
    \hline \\ [-2ex]
    US Pol. Books \cite{political-books} & Recommendation & 105 & 442 \\
    GR-QC \cite{snapnets} & Collaboration& 5242 & 14496 \\
    Wiki-Vote \cite{snapnets} & Voting & 7115 & 103689 \\
    NetPHY \cite{Chen09_2} & Collaboration & 37149 & 180826 \\
    Slashdot \cite{snapnets} & Social & 82168 & 948464 \\
    Twitter \cite{Kwak10} & Social & 41M & 1.5B
  \end{tabular}
  \caption[]{\label{tab:nets} Networks used in our experiments.}
\end{table}

\subsection{Comparison to the Exact IPs}\label{sec:ilp-runtime}
Since both T-SAA and E-SAA produce exact results for influence
maximization, we use them to show the optimality of \ISSA{}. As the SAA-based methods lack scalability, we run on the 105-node US Political
Books network only.
%Note that in this case we are averaging over 10 runs, generating different sample sets each time.
As shown in Fig. \ref{fig:ilp-comparison}, \ISSA{}
consistently performs as well as T-SAA while performing more than ten
times faster even despite the threading difference. Note that the decreasing runtime of \ISSA{} with increasing $k$ is on the scale of 5-10 seconds and can easily be explained as variation in the number of samples needed by \ISSA{} and in running time of the IP solver. Different behaviors are observed in running time for T-SAA and E-SAA between normalized and random costs. The reason is that the sampling procedure in T-SAA and E-SAA do not take the node benefit in the account. In contrast, the sampling algorithm in \ISSA{} will sample hyperedges starting from high-benefit nodes more often. Thus, the \ISSA{} running time is more stable across cost settings. 

%Going forward, we exploit the optimal behavior of \ISSA{} to place an upper bound on the performance of approximation methods in the next section.

\begin{figure}[ht]
  \centering
  \begin{subfigure}[t]{0.22\textwidth}
    \includegraphics[width=0.98\textwidth]{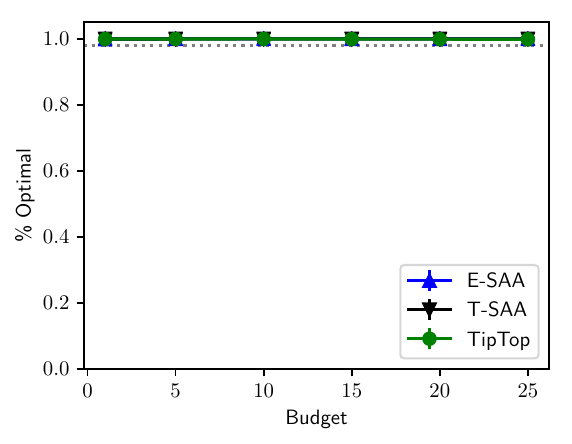}
    \caption{\label{fig:ilp-quality} Solution Quality (Normalized Linear costs)}
  \end{subfigure}
  \hspace{0.005\textwidth}
  \begin{subfigure}[t]{0.22\textwidth}
    \includegraphics[width=0.98\textwidth]{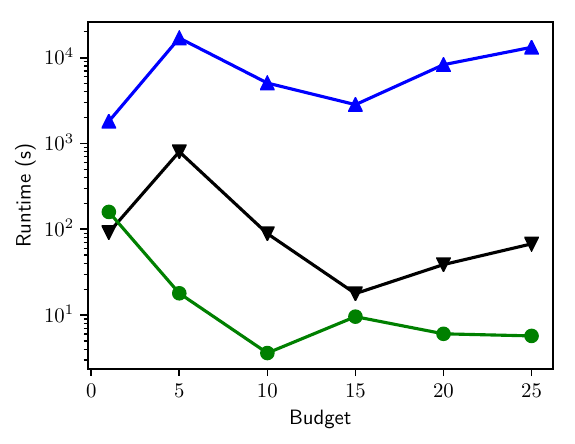}
    \caption{\label{fig:ilp-runtime} Running Time (Normalized Linear costs)}
  \end{subfigure}
  \hspace{0.005\textwidth}
  \begin{subfigure}[t]{0.22\textwidth}
    \includegraphics[width=0.98\textwidth]{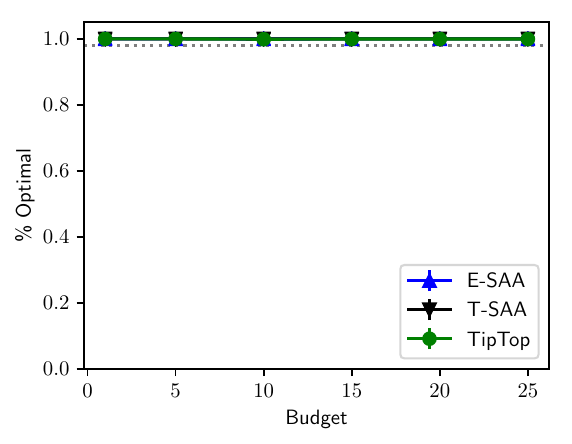}
    \caption{\label{fig:ilp-quality-random}Solution Quality (Random costs)}
  \end{subfigure}
  \hspace{0.005\textwidth}
  \begin{subfigure}[t]{0.22\textwidth}
    \includegraphics[width=0.98\textwidth]{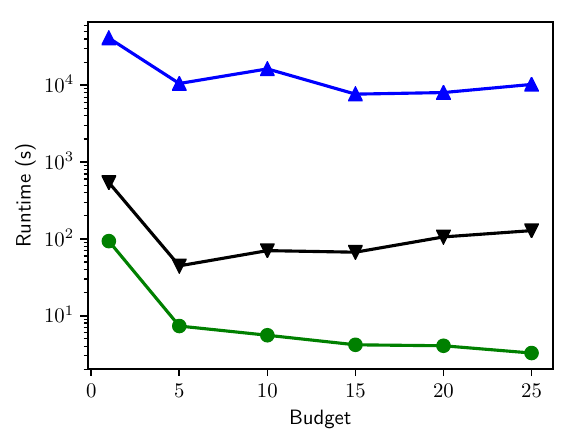}
    \caption{\label{fig:ilp-runtime-random}Running Time (Random costs)}
  \end{subfigure}
  \caption{\label{fig:ilp-comparison} The performance of \ISSA{}
    ($\epsilon = 0.02, \delta = 1/n$), T-SAA, and E-SAA ($T = 5000$) on the
    US Political Books\cite{political-books} network under the
    cost-aware problem setting. Costs are (\subref{fig:ilp-quality}-\subref{fig:ilp-runtime}) normalized by node degree and (\subref{fig:ilp-quality-random}-\subref{fig:ilp-runtime-random}) assigned uniformly at random on $[0,1)$. The $y$-axes of Figures (\subref{fig:ilp-runtime}), (\subref{fig:ilp-runtime-random}) are log-scaled.}
\end{figure}

\subsection{Benchmarking Greedy Methods}\label{sec:compare-greedy}

Having established that \ISSA{} is capable of producing almost exact solutions significantly faster than other IP-based solutions, we now exploit this property to place an upper bound on the performance of other algorithms. Ultimately, this allows us to make statements about the \textit{absolute} performance of these algorithms rather than merely their \textit{relative} performance. The algorithms we examine are IMM, BCT, and SSA under three problem settings across four networks.
% As shown in the previous section, \ISSA{} is capable of producing
% almost exact solutions significantly faster than other IP-based
% solutions. This property allows us to place an upper bound on the
% performance of other approximation methods, which in turn allows us to
% compare the \textit{absolute} performance of approximation methods
% rather than merely their \textit{relative} performance. We evaluate
% the performance of IMM, BCT, and SSA under three
% problem settings on three networks. 
All evaluations are under the IC model with edge probabilities set to $p(u,v) = 1/d_\text{in}(v)$, where $d_\text{in}(v)$ is the in-degree of $v$. This weight setting is adopted from prior work \cite{Tang15,Nguyen163}.

We first consider the traditional \IM{} problem, referred to as \textit{Unweighted} here to distinguish it from the subsequent problems. The authors' implementations of each algorithm are applied directly. Fig. \ref{fig:approx-perf-unweighted} shows that the decade-or-so of work on this problem has resulted in greedy solutions that far exceed their guarantee of $1 - 1/e - \epsilon$. Interestingly, this pattern continues under the \textit{Cost-Aware} setting (Fig. \ref{fig:approx-perf-cost}). 
% The first setting we consider is the traditional \IM{} problem
% (which we term \textit{Unweighted} to differentiate it from subsequent
% settings). We directly apply each author's implementation to this
% problem with no modifications. Fig. \ref{fig:approx-perf-unweighted}
% shows that the decade-or-so of work on this problem has resulted in
% greedy solutions that are much better than their approximation ratio of $(1 - 1/e - \epsilon)$. Somewhat surprisingly, we
% find similar results on the \textit{Cost-Aware} problem setting (fig. \ref{fig:approx-perf-cost}).

\textbf{Parameter Settings}. The Cost-Aware setting generalizes the Unweighted setting by adding a cost to
each node on the network.
In a social network setting, the costs can be understood
as the relative price each user sets for their participation in the marketing campaign. Note that neither IMM nor SSA support costs natively. We extend both to this problem by scaling their objective functions by cost and limiting the number of selected nodes by the sum of costs instead of the number of nodes, which gives each an approximation guarantee of $1 - 1/\sqrt{e} - \epsilon$ \cite{Khuller1999}. We consider three ways users may determine their costs.

\textit{Random Costs.} First, users may determine the value of their
influence independently from the network. We model this with Random Costs: each node selects a cost at random on $[0,1)$.

\textit{Degree-Normalized Costs.}
Users may instead use the most readily-available metric of their relative importance to determine prices: follower count. We examine two different functions of this, \textbf{Normalized Linear Costs} and \textbf{Normalized Logarithmic Costs}, to develop an understanding of how this impacts algorithm performance. In the linear case, each user sets its cost as $\text{cost}(u) = \frac{n}{|E|} d_\text{out}(u)$ ($d_\text{out}(u)$ is
the out-degree of $u$). The logarithmic case wraps the degree component: $\text{cost}(u) = \frac{n}{|E|} \log(d_\text{out}(u))$. The $n / |E|$ term normalizes costs across networks to allow using the same budget on each. These two cases function as rough upper and lower bounds on reasonable behavior, as shown by the following example. 

\begin{figure}[t]
  \centering
  \begin{subfigure}[t]{0.22\textwidth}
    \includegraphics[width=0.98\textwidth]{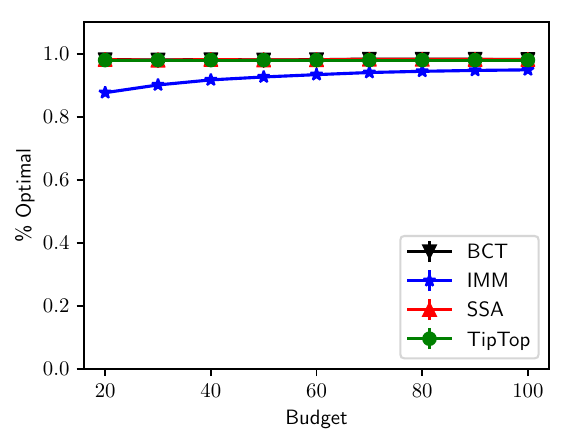}
    \caption{\label{fig:approx-perf-unweighted} Unweighted NetPHY }
  \end{subfigure}
  \hspace{0.005\textwidth}
  \begin{subfigure}[t]{0.22\textwidth}
    \includegraphics[width=0.98\textwidth]{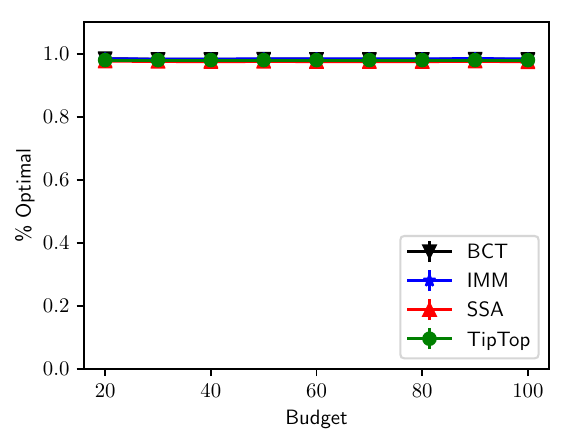}
    \caption{\label{fig:approx-perf-cost} Cost-Aware NetPHY }
  \end{subfigure}
%     \begin{subfigure}[t]{0.22\textwidth}
%     \includegraphics[width=0.98\textwidth]{figures/experiments/unweighted_soc-Slashdot0902_influence_k}
%     \caption{\label{fig:approx-perf-unweighted-slashdot} Unweighted Slashdot }
%   \end{subfigure}
%   \hspace{0.005\textwidth}
%   \begin{subfigure}[t]{0.22\textwidth}
%     \includegraphics[width=0.98\textwidth]{figures/experiments/5hot_cost-outdegree_benefit-0_9-uniform_soc-Slashdot0902_influence_k}
%     \caption{\label{fig:approx-perf-cost-slashdot} CTVM Slashdot }
%   \end{subfigure}
  \caption{\label{fig:approx-perf-prior} Mean performance of each
    approximation algorithm as the budget is varied under the
    Unweighted and Cost-Aware problem settings with $\epsilon = 0.02$.
    $OPT$ is estimated assuming that \ISSA{} achieves exactly
    $(1 - \epsilon) OPT$.}
\end{figure}

Suppose two users are setting prices, one with two thousand followers and one with two million. In the linear case if the former user demands \$800 to become an influencer, then the latter user will demand \$800,000. On the other hand, in the log case the latter will demand only \$1,527 -- slightly less than double that demanded by the two-thousand-follower user. We find that the performance remains similar to the unweighted case under degree scaling (see Fig. \ref{fig:approx-perf-prior}).

% Alternately, users may use the most readily-available metric of local
% network topology to determine the value of their influence: follower
% count. We consider both linear- and log-scaled functions of this
% metric as rough upper and lower bounds. 
% In the linear
% case ($\text{cost}(u) = n / |E| d_\text{out}(u)$; $d_\text{out}(u)$ is
% the out-degree of $u$, and the $n / |E|$ term normalizes the cost to
% allow each experiment to use the same budgets), if a user with two thousand
% followers requires \$800 to be influenced, then a user with two million
% followers would require \$800K. However, due to cascades increasing
% the effective influence of users with follower counts between the two,
% the high-follower users may opt to reduce their prices. The log-scaled
% function models an extreme case, with the two-million-follower user
% requiring a mere \$1,527. We show that the performance of all methods remains similar when scaling is changed, though the performance of IMM and SSA drops significantly when costs are randomized.

\textit{Benefit Settings.} Lastly, we consider the full \textbf{CTVM} problem described above. We
target 5\% of each network at random, assigning each targeted user a
\textit{benefit} on $[0.1, 1)$ and each non-targeted user a benefit of $0$. Costs are assigned according to one of the previous models (Random, Degree-Normalized Linear or Logarithmic). Neither IMM nor SSA can be extended to this problem without significant modifications, and therefore neither incorporates benefits.
Fig. \ref{fig:approx-perf-ctvm} shows the performance on the GR-QC,
Wiki-Vote, NetPHY and Slashdot networks.

\begin{figure*}[htb]
  \centering
  \begin{subfigure}[t]{0.22\textwidth}
    \includegraphics[width=0.98\textwidth]{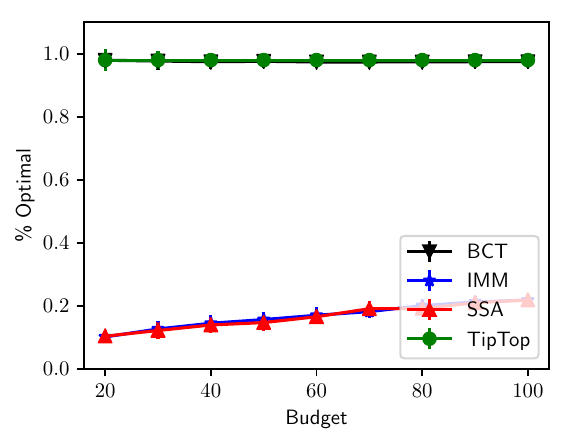}
    \caption{ GR-QC }
  \end{subfigure}
  \begin{subfigure}[t]{0.22\textwidth}
    \includegraphics[width=0.98\textwidth]{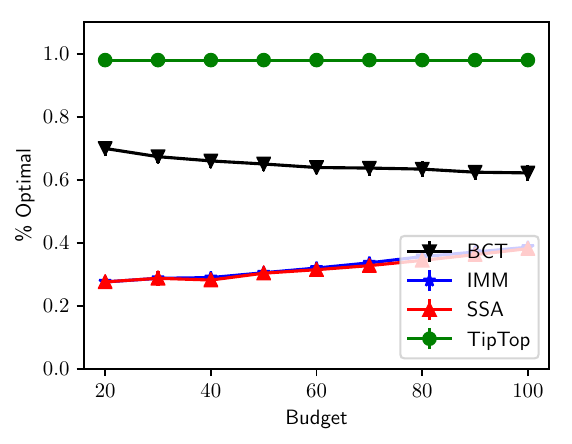}
    \caption{\label{fig:approx-ctvm-wiki-Vote} Wiki-Vote}
  \end{subfigure}
  \begin{subfigure}[t]{0.22\textwidth}
    \includegraphics[width=0.98\textwidth]{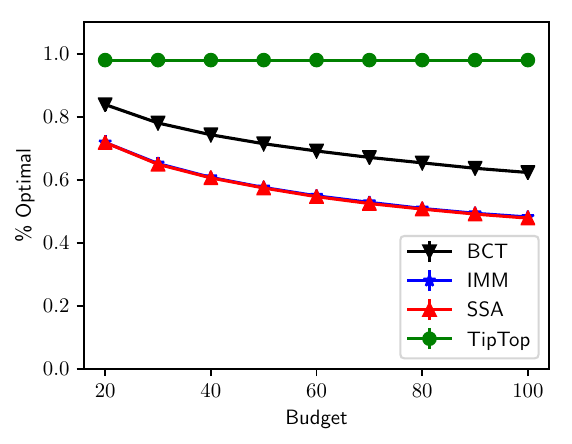}
    \caption{ \label{fig:approx-ctvm-phy} NetPHY }
  \end{subfigure}
  \begin{subfigure}[t]{0.22\textwidth}
    \includegraphics[width=0.98\textwidth]{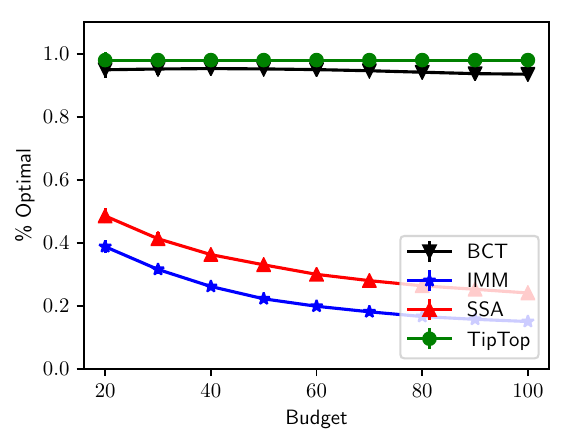}
    \caption{ \label{fig:approx-ctvm-slashdot} Slashdot}
  \end{subfigure}
  \caption{\label{fig:approx-perf-ctvm} Mean performance as the budget is varied under the CTVM
    problem setting with linear costs. Note that IMM and SSA have no guarantees on the CTVM setting.}
\end{figure*}
\begin{figure}
  \hspace{0.005\textwidth}
  \centering
%   \begin{subfigure}[t]{0.22\textwidth}
%     \includegraphics[width=0.98\textwidth]{figures/experiments/perf_ctvm-random_phy.pdf}
%     \caption{\label{fig:altcosts-a} Random -- NetPHY}
%   \end{subfigure}
%   \begin{subfigure}[t]{0.22\textwidth}
%     \includegraphics[width=0.98\textwidth]{figures/experiments/perf_ctvm-log-outdegree_phy.pdf}
%     \caption{\label{fig:altcost-log} Log-Outdegree -- NetPHY}
%   \end{subfigure}
    \begin{subfigure}[t]{0.22\textwidth}
    \includegraphics[width=0.98\textwidth]{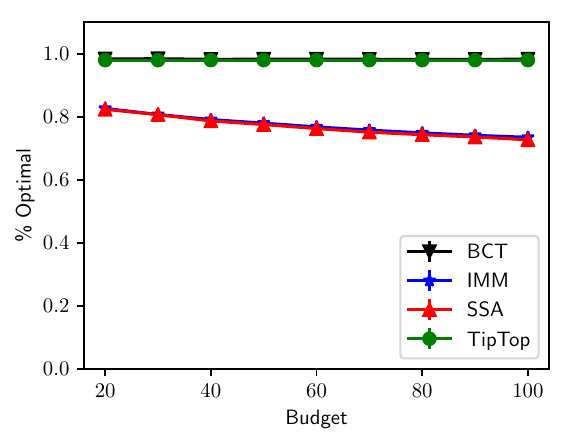}
    \caption{\label{fig:altcost-random} Random}
  \end{subfigure}
  \begin{subfigure}[t]{0.22\textwidth}
    \includegraphics[width=0.98\textwidth]{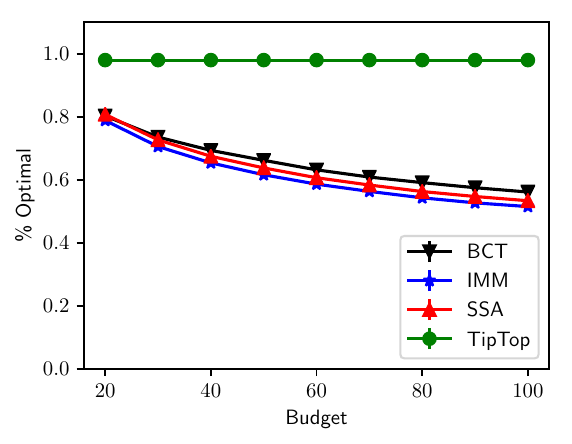}
    \caption{\label{fig:altcost-log} Log-Outdegree}
  \end{subfigure}

  \caption{\label{fig:altcosts} Performance on the Slashdot network under the CTVM setting with alternate costs.}
\end{figure}

From Fig. \ref{fig:approx-perf-ctvm}, we can see that the topology
has a significant impact on the performance of each algorithm. 
%Note that on GR-QC and NetPHY, the performance of BCT relative to the optimal tends downward while on Wiki-Vote it tends upward. 
Further, on NetPHY both IMM and SSA do astonishingly well despite being ignorant
of the targeted nodes. However, their behavior remains inconsistent across datasets due to their ignorance of the varying benefit assignments. Figure \ref{fig:altcosts} shows that BCT
performs similarly regardless of cost function, though random costs can cause notably worsened performance (Fig \ref{fig:altcost-random}).

\subsection{Analyzing \ISSA{}'s performance}\label{sec:tiptop-perf}

We now turn our attention %from establishing the performance of other
%approximation methods
to dissecting the performance of \ISSA{} in more detailed. Even it can be seen from the above section (Fig. 4 and 5) that \ISSA{} outperforms existing solutions, we did not examine the relative costs of running
\ISSA{} on these or larger networks.

\begin{figure}[ht]
  \centering
  \begin{subfigure}[t]{0.22\textwidth}
    \includegraphics[width=0.98\textwidth]{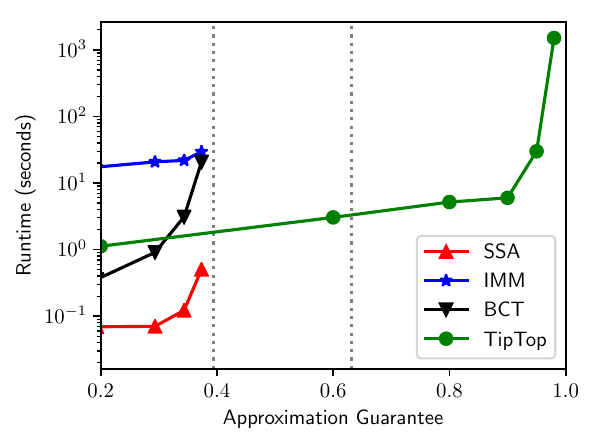}
    \caption{\label{fig:approx-runtime-eps} Runtime as $\epsilon$ is increased. Budget is held constant at $50$. Dotted lines show $1 - 1/\sqrt{e}$ and $1-1/e$, respectively.}
  \end{subfigure}
  \hspace{0.005\textwidth}
  \begin{subfigure}[t]{0.22\textwidth}
    \includegraphics[width=0.98\textwidth]{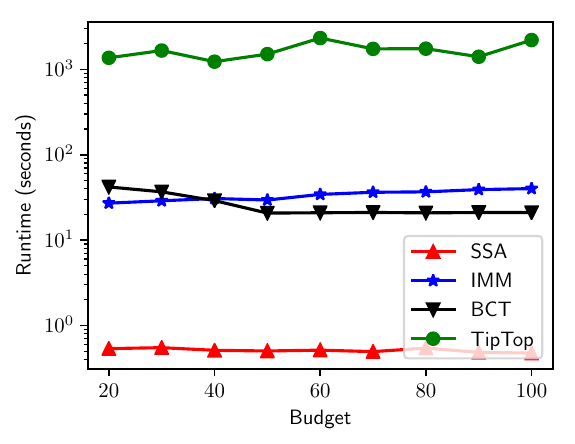}
    \caption{\label{fig:approx-runtime-k} Runtime as budget is increased. $\epsilon$ is held constant at $0.02$.}
  \end{subfigure}
  \caption{\label{fig:approx-runtime} Mean running time on the NetPHY network under the CTVM problem setting as $\epsilon$ and $k$ are varied.}
\end{figure}

{\bf Running Time.}
% Fig. \ref{fig:approx-runtime} compares the running time of each
% method as the approximation guarantee and budget are varied. From
% this, we can see that the runtime performance of \ISSA{} is not far from -- and in some cases better than -- greedy methods. With
% a comparable approximation guarantee, \ISSA{} is as fast as SSA and an
% order of magnitude faster than IMM. However, as the approximation
% guarantee tends towards 1, the advantage is lost. Fig.
% \ref{fig:approx-runtime-k} shows that while the runtime of both IMM
% and BCT is smoothly linear and for SSA is near-constant, the running
% time of \ISSA{} is less predictable. This is the cost of solving an IP
% instead of greedily solving MC: the IP solver may quickly
% find the optimal solution, or it may take exponentially long to do so. However, on
% average the running time of \ISSA{} is low enough that it is practical
% to run even on very large networks. To further establish the
% scalability of \ISSA{}, we run it on the Twitter dataset (fig.
% \ref{fig:twitter}). While the running time of \ISSA{} is worse, it still produces 98\% optimal solutions \textit{in under three hours}.
From Fig. \ref{fig:approx-runtime} we can see that the runtime performance of \ISSA{} is near that of the greedy algorithms. With a comparable approximation guarantee to the greedy methods, we can see that \ISSA{} has truly competitive runtime performance. However, as $\epsilon$ approaches 0 the runtime rapidly increases. Fig. \ref{fig:approx-runtime-k} illuminates one of the key costs of using an ILP: unpredictability. An ILP solver will in many cases find the solution relatively quickly, but in the worst case the complexity remains exponential. Finally, we note that in our tests on the Twitter dataset (Fig. \ref{fig:twitter}), we are able to solve the unweighted setting with a guarantee of 98\% in 13 hours and 28 minutes of CPU time -- which, as sampling is easily parallelized, takes \textit{under an hour} with 16 threads on our server. Further, for $\epsilon \geq 0.05$ the runtime of \ISSA{} remains competitive with and even outperforms other methods on the same dataset.

\begin{figure}[ht]
  \centering
  \begin{subfigure}[t]{0.22\textwidth}
    \includegraphics[width=0.98\textwidth]{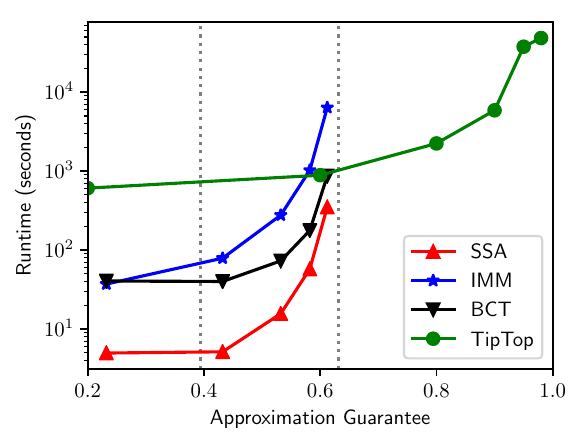}
    \caption{\label{fig:twitter-runtime} Running Time}
  \end{subfigure}
  \hspace{0.005\textwidth}
%   \begin{subfigure}[t]{0.22\textwidth}
%     \includegraphics[width=0.98\textwidth]{figures/experiments/perf_twitter_big_1_epsilon.pdf}
%    \caption{\label{fig:twitter-quality} Solution Quality}
%   \end{subfigure}
  \begin{subfigure}[t]{0.22\textwidth}
    \includegraphics[width=0.98\textwidth]{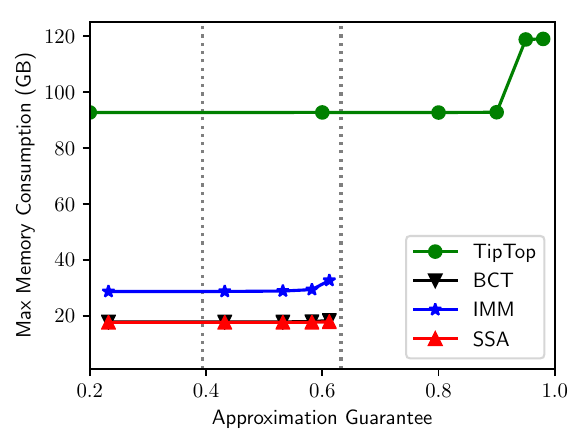}
   \caption{\label{fig:twitter-memory} Memory}
  \end{subfigure}
  \caption{\label{fig:twitter} Running time and memory consumption
    on the Twitter network under the unweighted problem setting. Only one repetition is used on this dataset. Budget is fixed at $50$. Dotted vertical lines show $1 - 1/\sqrt{e}$ and $1 - 1/e$, respectively.}
\end{figure}

{\bf Number of Samples.} 
% Previously, solving an IP on a network with millions
% of nodes and billions of edges had been thought impractical. \ISSA{}
% is able to accomplish this by dramatically reducing the number of
% samples used to solve MC. Table \ref{tab:samples} shows that on
% average, \ISSA{} uses on the order of $10^3$ fewer sames than greedy
% methods to solve MC. As SSA has a Stare phase in which it verifies the quality of a generating solution, we also compare the number of samples used in our method and SSA. The lower section of table \ref{tab:samples} shows that even
% then, the number of samples used to verify the approximation guarantee
% is similar between the greedy method and \ISSA{}. The total number of samples used in both phases of \ISSA{}, in general, in smaller that that of SSA. Further, the
% approximation ratio of \ISSA{} affords greater precision if one needs
% to guarantee a certain level of performance, as shown in Fig. 7. Where
% prior
% work %could only make guarantees less than or equal to $(1 - 1/e - \epsilon)$,
% produce results of similar quality regardless of $\epsilon$, \ISSA{}
% enables setting guarantees to exactly the required value. If, for
% instance, an application requires an 80\% approximation ratio,
% $\epsilon$ can be set to exactly $0.2$, at which point \ISSA{} is as
% fast as state-of-the-art greedy solutions.
The problem of maximizing influence on billion-scale graphs like Twitter is incredibly difficult -- a fact further complicated by the worst-case complexity of solving an IP. \ISSA{} addresses this by dramatically reducing the number of samples needed to solve the problem. Table \ref{tab:samples} shows that on average, \ISSA{} uses $10^3$ fewer samples than prior work. Each sample corresponds to one constraint in the IP formulation (specifically, Eqn. \ref{eqn:sample-constraint}), and therefore this reduction has a direct impact on the size of the resulting IP.

This is made possible by a SSA-like approach: adding a ``Stare'' phase to the algorithm. This phase validates the quality of the solution with significantly more samples than the IP uses. Since verification does not require an IP, it has similar complexity to the Stare phase in SSA. The samples used in each algorithm are divided into two sections in Table \ref{tab:samples}. The upper section corresponds to samples used in generating the solution, while the bottom section corresponds to samples used in the Stare phases of SSA and \ISSA{}. Note, however, that \ISSA{} takes a similar number of samples for verification to SSA when obtaining a similar approximation guarantee while also using dramatically fewer for solving the IP.
%The end result of this is that the approximation quality of \ISSA{} has a greater correlation with the choice of $\epsilon$ than other methods (see Fig. \ref{fig:twitter}).
%Note, however, that even when $\epsilon = 0.8$, \ISSA{} still performs in excess of its bounds by more than 60\%.

\begin{table}[ht]
  \centering\small
  \begin{tabular}{l|rrr}
\toprule
Coverage        & Unweighted      & Cost-Aware      & CTVM            \\
\midrule                          
IMM             & \num{38492250}  & \num{7682692}   & \num{7682692}   \\
SSA             & \num{236139}    & \num{517892}    & \num{91538}     \\
BCT             & \num{240567072} & \num{2260952}   & \num{1130476}   \\
\textsc{TipTop} & \num{12145}     & \num{60155}     & \num{3658}      \\
\midrule                          
Verification    & Unweighted      & Cost-Aware      & CTVM            \\
\midrule                          
SSA             & \num{5382957.2} & \num{1072126.5} & \num{1149159.4} \\
\textsc{TipTop} & \num{355000}    & \num{52000}     & \num{29000}     \\
\bottomrule
\end{tabular}

  \caption{\label{tab:samples} Mean required samples for each
    algorithm on NetPHY. Approximation guarantees are $~0.61$ for greedy methods
    in the unweighted case and $~0.37$ for the cost-aware case (and
    the CTVM case for BCT). The approximation guarantee for \ISSA{} is
    $0.6$ in all cases. \textit{Coverage}: Samples input into the
    MC solver. \textit{Verification}: Samples used to verify
    solution quality. IMM and BCT do not incorporate a verification
    stage.}
\end{table}

\textbf{Memory Consumption.} Lastly, we briefly examine the amount of memory consumed by each algorithm. We measure this by running each algorithm with the \texttt{time} binary\footnote{Note that this is distinct from the bash built-in command.} present on Debian, which reports peak memory usage in kilobytes. While for the greedy methods, the number of samples used is the dominating factor, the IP solver used in \ISSA{} may consume an additional large chunk of memory. Figure \ref{fig:twitter-memory} shows that while this may be true, it does not consume an unreasonable amount of memory.  \ISSA{} peaks at 120GB on the Twitter network---a value which we find wholly reasonable for solving with an error of $1 - \epsilon$ on a billion-scale network.

\vspace{-0.1in}
	\section{Conclusion}
    In this paper, we propose the first (almost) exact and optimal solutions to the \CTVM{} (and thus \IM{}) problem, namely T-EXACT, E-EXACT, and \ISSA{}. T-EXACT and E-EXACT use the traditional stochastic programming approach and thus suffer the scalability issue. In contrast, our \ISSA{} with innovative techniques in reducing the number of samples to meet the requirement of ILP solver is able to run on billion-scale OSNs such as Twitter. \ISSA{} has an approximation ratio of $(1-\epsilon)$ which significantly improves from the current best ratio $(1-1/e-\epsilon)$ for \IM{} and $(1-1/\sqrt{e}-\epsilon)$ for \CTVM{}. \ISSA{} also lends a tool to benchmark the absolute performance of existing algorithms on large-scale networks.
%	In this paper, we propose the \CTVM{} problem that generalizes several viral marketing problems including the classical \IM. We propose \CTIM{} an efficient approximation algorithm to solve \CTVM{} in billion-scale networks and show that  it is both theoretically sound and practical for large networks. The algorithm can be employed to discover more practical solutions for viral marketing problems, as illustrated through the discovering of influential users w.r.t. trending topics in Twitter media site.

\section*{Acknowledgement}
This work was partially supported by the NSF CNS-1443905, NSF EFRI 1441231, and DTRA HDTRA1-14-1-0055.
	\bibliographystyle{IEEEtran}
	\bibliography{social,pids,targetedIM,budgetedIM,vul,cnd}

\begin{IEEEbiography}[{\includegraphics[width=1in,height=1.25in,clip,keepaspectratio]{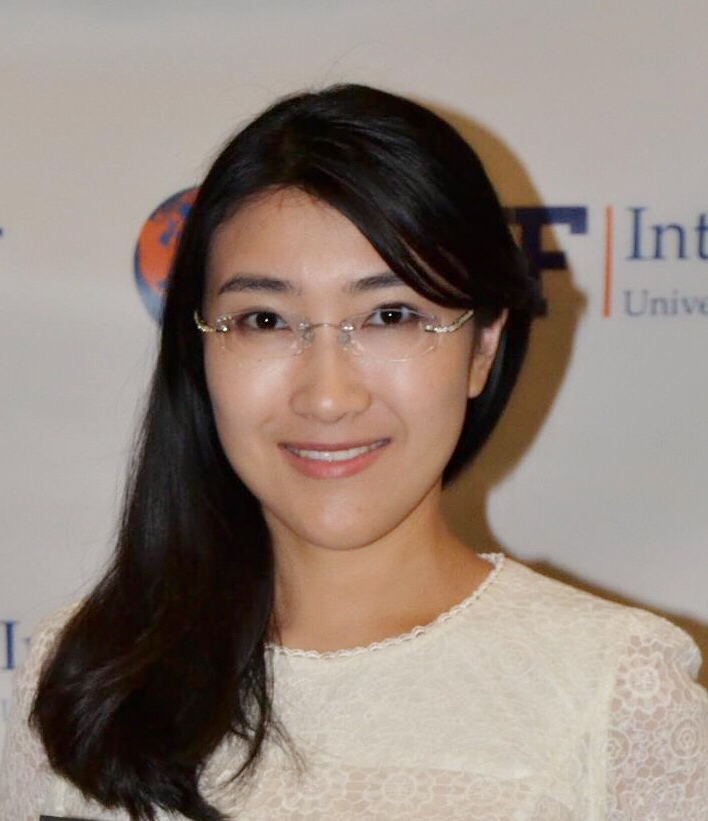}}]{Xiang Li}(M'18)
received her Ph.D. degree in Computer and Information Science and Engineering department of University of Florida. She is an Assistant Professor at the Department of Computer Engineering of Santa Clara University. Her research interests are centered on the large-scale optimization and its intersection with cyber-security of networking systems, big data analysis, and cyber physical systems.  She has published 25 articles in various prestigious journals and conferences such as IEEE Transactions on Mobile Computing, IEEE Transactions on Smart Grids, IEEE INFOCOM, IEEE ICDM, including one Best Paper Award in IEEE MSN 2014, Best Paper Nominee in IEEE ICDCS 2017, and Best Paper Award in IEEE International Symposium on Security
and Privacy in Social Networks and Big Data 2018.

She has served as Publicity Co-Chair of International Conference on Computational Data \& Social Networks 2018, Session Chair of ACM SIGMETRICS International Workshop 2018, and on TPC of many conference including IEEE ICDCS, IEEE ICDM workshop, COCOA etc., and also served as a reviewer for several journals such as IEEE Transactions on Mobile Computing, IEEE Transactions on Networks Science and Engineering and Journal of Combinatorial Optimization, etc. 
\end{IEEEbiography}

\begin{IEEEbiography}[{\includegraphics[width=1in,height=1.25in,clip,keepaspectratio]{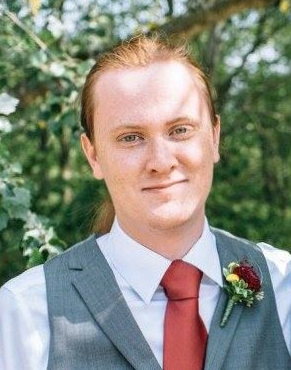}}]{J. David Smith}
received his Bachelor of Science degree from the University of Kentucky. He is currently working towards his PhD in Computer Science at CISE Department, University of Florida, USA under the supervision of Dr. My T. Thai. His current research interests are in security on online social networks.
\end{IEEEbiography}

%\newpage

\begin{IEEEbiography}[{\includegraphics[width=1in,height=1.25in,clip,keepaspectratio]{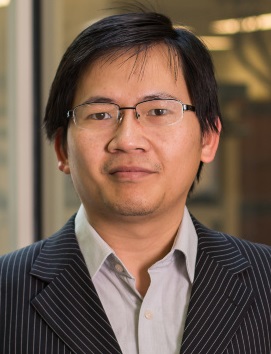}}]{Thang N. Dinh} (S'11-M'14) received the Ph.D. degree in computer engineering from the University of Florida, in 2013. He is currently an Assistant Professor with the Department of Computer Science, Virginia Commonwealth University. His research focuses on security and optimization challenges in complex systems, especially blockchain, social networks,  critical infrastructures. His research won several best papers and best papers nominees from IEEE ICDM, ACM SenSys, and CSoNet.  He served as the PC co-chair for the BlockSEA'18, COCOON'16, CSoNet'14 and the TPC of several conferences including the IEEE INFOCOM, the ICC, the GLOBECOM. He is an Assoc. Editor for Springer Nature, Journal on Comp. Social Sci. and Rev. Editor for Big Data.
\end{IEEEbiography}

\begin{IEEEbiography}[{\includegraphics[width=1in,height=1.25in,clip,keepaspectratio]{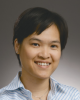}}] {My T. Thai}(M'06) is a UF Research Foundation Professor at the Computer and Information Science and Engineering department, University of Florida. Her current research interests are on scalable algorithms, big data analysis, cybersecurity, and optimization in network science and engineering, including communication networks, smart grids, social networks, and their interdependency. The results of her work have led to 6 books and 150+ articles, including IEEE MSN 2014 Best Paper Award, 2017 IEEE ICDM Best Papers Award, 2017 IEEE ICDCS Best Paper Nominee, and 2018 IEEE/ACM ASONAM Best Paper Runner up.

Prof. Thai has engaged in many professional activities. She has been a TPC-chair for many IEEE conferences, has served as an associate editor for Journal of Combinatorial Optimization (JOCO), Journal of Discrete Mathematics, IEEE Transactions on Parallel and Distributed Systems, IEEE Transactions on Network Science and Engineering, and a series editor of Springer Briefs in Optimization. She is a founding EiC of the Computational Social Networks journal. She has received many research awards including a UF Provosts Excellence Award for Assistant Professors, UFRF Professorship Award, a Department of Defense (DoD) Young Investigator Award, and an NSF (National Science Foundation) CAREER Award.
\end{IEEEbiography}
\end{document}